\documentclass[aps,pra,twocolumn,longbibliography,superscriptaddress]{revtex4-1}
\usepackage{amsmath}
\usepackage{amsthm}
\usepackage{latexsym}
\usepackage{amssymb}
\usepackage{bm}
\usepackage{bbm}
\usepackage{appendix}
\usepackage{graphics,epstopdf}
\usepackage{physics}
\usepackage{color}
\usepackage[colorlinks=true,linkcolor=blue,citecolor=green,plainpages=false,pdfpagelabels]{hyperref}
\usepackage[normalem]{ulem}
\usepackage{newlfont}
\usepackage{amsfonts}
\usepackage{amsthm}
\usepackage{times}
\usepackage[T1]{fontenc}

\usepackage{graphicx}
\usepackage{epsfig}
\usepackage{mathrsfs}
\usepackage[thinc]{esdiff}
\usepackage{caption}
\usepackage{subcaption}

\usepackage{times}

\newtheorem{corollary}{Corollary}
\def\tr{\operatorname{tr}}
\newtheorem{theorem}{Theorem}

\DeclareOldFontCommand{\rm}{\normalfont\rmfamily}{\mathrm}

\DeclareOldFontCommand{\rm}{\normalfont\rmfamily}{\mathrm}

\usepackage{bigints}

\begin{document}

\title{Quantum Speed limit on the production of quantumness of observables}

\author{Divyansh Shrimali}\email{divyanshshrimali@hri.res.in}
\affiliation{Harish-Chandra Research Institute,\\  A CI of Homi Bhabha National
Institute, Chhatnag Road, Jhunsi, Prayagraj 211019, India
}
\author{Swapnil Bhowmick}\email{swapnilbhowmick@hri.res.in}
\affiliation{Harish-Chandra Research Institute,\\  A CI of Homi Bhabha National
Institute, Chhatnag Road, Jhunsi, Prayagraj 211019, India}

\author{Arun Kumar Pati}\email{arun.pati@tcgcrest.org}
\affiliation{Centre for Quantum Engineering, Research and Education (CQuERE), TCG CREST, Kolkata, India}

\begin{abstract}
Non-classical features of quantum systems can degrade when subjected to environment and noise. Here, we ask a fundamental question: What is the minimum amount of time it takes for a quantum system to exhibit non-classical features in the presence of noise?
Here, we prove distinct speed limits on the quantumness of observable as the norm of the commutator of two given observables. The speed limit on such quantumness measures sets the fundamental upper bound on the rate of change of quantumness, which provides the lower bound on the time required to change the quantumness of a system by a given amount. Additionally, we have proved speed limit for the non-classical features such as quantum coherence
that captures the amount of superposition in the quantum systems.
We have demonstrated that obtained speed limits are attainable for physical processes of interest, and hence, these bounds can be considered to be tight.




\end{abstract}

\maketitle

\section{Introduction}


Quantumness is among a unique set of non classical properties of quantum mechanics that sets them apart from classical notions in physics. These properties include, but are not limited to superposition, entanglement, and non-commutativity, which play their crucial role in the behaviour of quantum systems. Intuitively, non-commutativity, coherence, and quantum correlations are all closely related concepts. Non-commutativity in particular is responsible for the evolution of quantum systems, as it determines how  the state and observables of a system change over time. In this sense, non-commutativity is more fundamental because it is an essential component in the quantum state evolution for generation of other quantum mechanical properties like superposition and entanglement. This property also gives rise to the uncertainty relation theorised first by Heisenberg~\cite{Heisenberg} and proven rigorously by Robertson~\cite{Robertson1929}, which essentially sits at the core of quantum mechanics. Several attempts have been made over the years in studying this concept and being reinterpreted like in~\cite{Pati-2007,Chen-2015, Chen-2016}, achieving better and stronger bounds as in~\cite{Pati-2014,Englert-2024, Mondal-2017}, has lead to several applications in recent years \cite{Shrimali2022, Hasegawa2020, Yu2022}. Furthermore, recent studies have demonstrated that the coherence of a system can be estimated by measuring the non-commutativity between an observable and its incoherent part (in reference basis)~\cite{Tanaya-22}.

Quantum physics imposes fundamental limitations on the rate at which the state and observables of quantum system evolve in time when quantum system is subjected to external field or environment. These limitations are known as quantum speed limits (QSLs). Quantum speed limits provide a lower bound on the time required for a change in the distinguishability of initial and time evolved quantum states or the expectation value of a given observable through a physical process. QSLs have been extensively studied for both closed system dynamics~\cite{Mandelstam1945,Margolus1998,Aharonov1961,Anandan1990}, open system dynamics~\cite{Campo2013}, non-Hermitian dynamics~\cite{Impens2021,Brody2019,Alipour2020,Uzdin2012,Dimpi2022} and many-body dynamics~\cite{Anthony-23, Zhang-23}. The stronger bounds on speed limits for states and observables~\cite{Shrimali2024} are also proposed and worked out. Due to their fundamental nature, QSLs have found applications in rapidly developing quantum technologies such as quantum control~\cite{Caneva-09, Campbell-17, Evangelakos-23}, quantum computation~\cite{Ashhab-12, Mohan-22, Aifer-22}, quantum communication~\cite{Wiesner-92}, quantum energy storage devices~\cite{Shrimali2024,Mohan2021}, and quantum thermal machines~\cite{Victor-23, Bhandari-22}.

In quantum physics, observables have an intriguing property of non-commutativity which we do not see in classical physics. 
Non-commutativity of two observables is often used to study several physical phenomena, such as the scrambling of information~\cite{Xu2024}, quantum chaos~\cite{Jensen1992}, quantum phase transitions~\cite{heyl2018,Vojta2003}, linear response theory~\cite{Vliet1979} and the generation of non-classicality(quantumness)~\cite{Jing2016}. Several witnesses of non-classicality and quantum informational measures have been defined using non-commutativity of observables such as quantumness~\cite{Jing2016}, skew information~\cite{Luo-03} and coherence~\cite{Hu-2017, Bu-2017}. We often encounter a situation where physical observables evolve and become incompatible with the initial observable (or given observable), as they do not commute.

The degree of non-commutativity between two observables is associated with how much off-diagonal part one observable has in the eigenbasis of the other observable. The non-commutativity of two observables can be quantified by the norm of their commutator. Moreover, it has been found that quantum coherence and non-commutativity are related concepts \cite{Tanaya-22}. In this work, we consider non-commutativity and quantum coherence as two distinct measures of quantumness. Given the importance of quantumness, a natural question arises: is it possible to formulate QSLs for quantumness. To address this question, we derive the quantum speed limit for the quantumness of observables and quantum coherence, where former is defined as the norm of the commutator between two given observables. The quantum speed limits on the quantumness provide a lower bounds on the timescale required to degrade a given degree of quantumness. This bound quantifies how fast an observable becomes incompatible with a stationary/initial observable or the extent to which the time-evolved observable is incompatible with the stationary observable under an arbitrary physical process. By establishing this bound, we can classify physical processes based on their ability to generate non-classicality in a minimal time, which has implications for the understanding of quantum dynamics. Furthermore, we have estimated speed limits on the quantumness of observables for an example case of markovian and pure dephasing dynamics. We found that pure dephasing dynamics saturate speed limits on the quantumness of observables. This finding provides new insights into the nature of pure dephasing dynamics and highlights its importance in generating non-classicality in the fastest possible way. Moreover, we have obtained speed limit bounds on the generation and degradation of skew information and coherence. In general we refer to the speed limits derived in this paper as speed limits on quantumness.

This paper is organised as follows. In section~\ref{sec:Prelims}, we discuss the preliminaries and background required to arrive at the main results of this paper. In section~\ref{sec:QSL_production}, we derive the speed limits on the quantumness of observables, the skew information and the coherence. In section~\ref{sec:Illustrations}, we estimates the speed limits on the quantumness for physical process of interest. Finally, in the last section~\ref{sec:Conclusion}, we provide conclusions.

\section{Preliminaries}\label{sec:Prelims}
In this section, we briefly review some of the standard notations and results common in the literature of quantum information theory.


{\it \bf Quantum dynamics:}  Let us consider a quantum system described by a density operator  $\rho \in \cal{D}(\cal{H})$ that follows an open quantum dynamics.
The evolution of the density operator in the Schr\"odinger picture can be described by a linear differential equation known as the master equation given by 

 \begin{equation}
      \dot{\rho_t}:=\diff{\rho_t}{t}=\mathcal{L}_t(\rho_t) \label{Master_equation_rho},
 \end{equation}
where $\rho_t $ is the state of the system at time $t$ and $\mathcal{L}_t$ is the Liouvillian super-operator~\cite{Rivas-12} which in general can be time independent or time-dependent.
If $\cal{L}_{t}$ is time independent, then the state at any time $t$ is given as $\rho_{t} = \operatorname{e}^{t\cal{L}}(\rho_{0})$.

Let us now see how observables evolve in the Heisenberg picture. The adjoint of a linear completely positive trace preserving (CPTP) map $\Phi$ is denoted by $\Phi^\dag:\cal{B}(\cal{H})\rightarrow\cal{B}(\cal{H})$ and is defined as the unique linear map that satisfies $\tr({\cal{A}\Phi(\rho)}) = \tr({\Phi^\dag}(\cal{A})\rho) $, $\forall \rho \in \cal{D}(\cal{H}), \cal{A} \in \cal{B}(\cal{H})$. Since $\Phi$ is trace preserving, $\Phi^{\dag}$ is unital, i.e., $\Phi(\mathbbm{1}_{d})=\mathbbm{1}_{d}$. 

The evolution of an observable $\mathcal{A}_t$ in the Heisenberg picture is given by the map ${\Phi^{\dag}}$  which keeps the density operator fixed and evolves an initial observable to a final observable. The only constraint on ${\Phi}^{\dag}$ is that it must be unital. In differential equation form the time evolution of an observable is given by adjoint-master equation:
\begin{equation}
     \dot{\cal{A}}_t:= \diff{ \cal{A}_t }{t} =\cal{L}_{t}^{\dagger}(\cal{A}_t)\label{Master_equation_observable},
\end{equation}
where $\cal{A}_t $ is the observable of the system at time $t$ and  $\cal{L}^{\dagger}_{t}$ is adjoint of the Liouvillian super-operator. If $\cal{L}_{t}$ is time independent the observable at time $t$ is given as $\cal{A}_{t} = \operatorname{e}^{t \cal{L}^{\dag}}(\cal{A}_{0})$.

{\it \bf Quantumness of Observables:}
Suppose $\mathcal{A}$ and $\mathcal{B}$ are two observables of a quantum system which is described by Hilbert space $\mathcal{H}$. We can use any suitable norm of the commutator of two observables to quantify the non-commutativity of observables \cite{Ma2014,Guo2016,Mord2022}. This is because $[\mathcal{A}, \mathcal{B}]= 0$ iff $\norm{[\mathcal{A}, \mathcal{B}]}=0 $ where $\norm{\cdot}$ represents a suitable norm defined on the operator space. 
For example, one has the Schatten-p norm of an operator $O\in \cal{B(H)}$ which is defines as
\begin{equation}
    \|O\|_{p}=(\tr |O|^{p})^{1/p},
\end{equation}
where $|O| = \sqrt{O^{\dagger}O}$,\, $p\geq 1,\, p\in \mathbb{R}$. The operator norm, the Hilbert-Schmidt norm, and the trace norm corresponds to $p=\infty,\cdots,2,1$, respectively and satisfy the inequality $\|A\|_{\rm op}\leq\|A\|_{\rm HS}\leq\|A\|_{\rm tr}$. In the sequel, we use the Hilbert-Schimidt norm to define the quantumness of observable.
Since $\norm{[\mathcal{A}, \mathcal{B}]}\neq0$  implies that $\mathcal{A}$ and $\mathcal{B}$ do not commute, the quantity $\norm{[\mathcal{A}, \mathcal{B}]}$ is used to define the quantumness of  $\mathcal{A}$ and $\mathcal{B}$ as follows
 \begin{equation}\label{Quantumness}
     Q(\mathcal{A},\mathcal{B})=2 \norm{[\mathcal{A}, \mathcal{B}]}_{\rm HS}^{2},
 \end{equation}
 where $\norm{O}_{\rm HS}=\sqrt{\tr(O^{\dagger}O)}$  is the Hilbert-Schmidt norm of $O$. If we choose $\mathcal{A}=\rho$ and $\mathcal{B}=\sigma$, where $\rho$ and $\sigma$ are density matrices, then $ Q(\rho,\sigma)$ reduces to the definition of quantumness for states first introduced in Refs.~\cite{Iyengar2013,Ferro2018}. It can be easily measured in the interferometric setups~\cite{Ferro2018}. In Ref.~\cite{Hu-2017}, it has been demonstrated that quantumness establishes the lower bound for the relative quantum coherence.
The quantumness of observables is related to the out of time ordered correlators (OTOC) \cite{Xu2024}, which is used in studying quantum chaos and scrambling in many-body physics. It has been established that the observable $\mathcal{O}=i[\mathcal{A},\mathcal{A}^{D}]$ serves as coherence witness, where $\mathcal{A}$ is unit-norm traceless observable and $\mathcal{A}^{D}$ is dephased observable of $\mathcal{A}$ in reference basis. The absolute value of 
the expectation value of observable $\mathcal{O}$ in a given pure state $\ket{\psi}$ is lower bounded by twice of $l$-1 norm of coherence of $\ket{\psi}$ in reference basis. Moreover, we can show that for qubit observables the quantity $\norm{[\mathcal{A}, \mathcal{B}]}_{\rm HS}$ is proportional to the $l_1$ norm of coherence of $\mathcal{B}'$ in basis of $\mathcal{A}$, where $\mathcal{B}'=\mathcal{B} + \lambda I$ and $-\lambda$ is smallest negative eigenvalue of $\mathcal{B}$.

{\it \bf Skew Information:} In the definition of quantumness~\eqref{Quantumness}, if we choose $\mathcal{A}=\sqrt{\rho}$ and $\mathcal{B}=A$, and multiply factor 1/4 then we obtain 
 \begin{equation}
     I(\rho,A)=\frac{1}{2} \norm{[\sqrt{\rho}, A]}_{\rm HS}^{2},
 \end{equation}
where $\rho$ and $\cal{A}$ are state and observable of the given quantum system, respectively.
This quantity is known as the Dyson-Wigner-Yanase skew information \cite{Luo-03}. It has many interpretation in quantum information theory. Skew information lower bounds the quantum uncertainty of the observable $\cal{A}$ in the quantum state $\rho$~\cite{Luo2006} and quantum fisher information~\cite{Luo2004}. Skew information has been used to construct measures of quantum correlations~\cite{Luo2012}, quantum coherence~\cite{Luo2022,Yu2017} and so on. It has also been widely used to study uncertainty relations~\cite{Luo-03}, quantum phase transitions~\cite{Li2016}, quantum speed limits~\cite{Marvian-2016, Pires-2016}, etc.
Skew information has quite different from, however it is deeply related to, the celebrated von Neumann entropy~\cite{Wei-2020}.

{\it \bf Quantum Coherence:} Quantum superposition is one of the remarkable features of the quantum theory and the amount of superposition present in a particular state is quantified by quantum coherence. Quantum coherence plays a pivotal role in numerous quantum information processing tasks~\cite{Streltsov2017}, quantum computation~\cite{Ahnefeld2022}, and quantum thermodynamics~\cite{Hammam2022,Gour2022}. In context of quantum thermodynamics, the energy eigenbasis becomes a natural reference basis of choice~\cite{Lostaglio2019}. In general, quantum coherence is a basis dependent quantity. There are several widely known quantum coherence measures such as the relative entropy of coherence and the $l_1$ norm of coherence~\cite{Baumgratz-14}, skew-information based coherence measure~\cite{Luo2022}, the geometric coherence~\cite{Xiong2018}, and the robustness of coherence~\cite{Napoli2016}, etc. We are using the skew-information based coherence measure because of its relation to non-commutativity of observables. In addition, it is also easier to work and compute compared to some other measures of coherence.

The skew-information based quantum coherence of a state $\rho$ in the basis of given observable $\cal{A}$, i.e., $\{\ket{k}\}$ can be quantified by \cite{Yu2017, Girolami-2014}
\begin{equation}
    C(\rho, \cal{A}) = \sum_{k=0}^{N_{D}-1}I(\rho, \ket{k}\bra{k}),
\end{equation}
where $I(\rho,\ket{k}\bra{k})=-\frac{1}{2}\Tr{[\sqrt{\rho},\ket{k}\bra{k}]}^{2}$ represents the skew information with respect to the projector $\ket{k}\bra{k}$. The 
above measure of quantum coherence is a strongly monotonic one.

\section{QSL on production of quantumness of Observables, skew information and coherence}\label{sec:QSL_production}
In this section we will present some distinct speed limits which will be based on change of quantumness of observables, skew information and coherence. 

In quantum physics, we often encounter situations where observables or states evolve in time and at any later time they fail to commute with initial observables or states. This is the fundamental feature of quantum dynamics. To this end, the natural question arises "how fast time evolved observables fail to commute with initial observables or what is the timescale when they become incompatible to each other?" To answer this question we have derive speed limit on the quantumness of observables.




\begin{theorem}\label{Th1}
For any given observable ($\mathcal{A}$) of a quantum system whose time evolution is governed by the adjoint of the Liouvillian super-operator $\mathcal{L}^{\dagger}_t$, the time required to generate a certain amount of quantumness $Q(\mathcal{A}_0, \mathcal{A}_T)$ is bounded below by:
\begin{equation}\label{QQSL}
    T \geq T_{Q} = \frac{\sqrt{Q(\mathcal{A}_0, \mathcal{A}_T)}}{\sqrt{2}\langle\!\langle \lVert [\mathcal{A}_0, \mathcal{L}^{\dagger}(\mathcal{A}_t)] \rVert_{\rm HS} \rangle\!\rangle}_{\rm T},
\end{equation}
where $\mathcal{A}_0$ and $\mathcal{A}_T$ represent the initial and final observables of the given quantum system, respectively. Additionally, $\langle\!\langle X_t \rangle\!\rangle_{T} \equiv \frac{1}{T} \int_{0}^{T} X_t \,dt$ denotes the time average of the quantity $X_t$.
\end{theorem}

\begin{proof}
   If we choose $\mathcal{A}_{X}=\mathcal{A}_{0}$ and $\mathcal{A}_{Y}=\mathcal{A}_{t}$ in~\eqref{Quantumness} such that $[\mathcal{A}_0, \mathcal{A}_t]\neq0$ (where $\mathcal{A}_{t}=e^{\mathcal{L}^{\dag}_t}(\mathcal{A}_{0})$) . The quantumness $Q(\mathcal{A}_0,\mathcal{A}_t)$ at time $t$ can be expressed as
\begin{equation}
Q(\mathcal{A}_0,\mathcal{A}_t) = 2\tr([\mathcal{A}_0,\mathcal{A}_t]^{\dagger}[\mathcal{A}_0,\mathcal{A}_t]).
\end{equation}
After differentiating the above equation with respect to time $t$, we obtain

\begin{align}
    \frac{{\rm d} }{{\rm d}t}Q(\mathcal{A}_0,\mathcal{A}_t) =\hspace{0.2cm} & 2\tr([\mathcal{A}_0,\dot{\mathcal{A}}_t]^{\dagger}[\mathcal{A}_0,\mathcal{A}_t])+\nonumber \\ & 2\tr([\mathcal{A}_0,\mathcal{A}_t]^{\dagger}[\mathcal{A}_0,\dot{\mathcal{A}_t}]).
\end{align}

Let us now consider the absolute value of the above equation and apply triangular inequality $|A+B|\leq |A|+|B|$. We then obtain the following inequality

\begin{align}
     \left|\frac{{\rm d} }{{\rm d}t}Q(\cal{A}_0,\cal{A}_t)\right| \leq\hspace{0.2cm}& 2\left|\tr([\cal{A}_0,\dot{\cal{A}}_t]^{\dagger}[\cal{A}_0,\cal{A}_t])\right|\nonumber \\&+2\left|\tr([\cal{A}_0,\cal{A}_t]^{\dagger}[\cal{A}_0,\dot{\cal{A}_t}])\right|.
\end{align}
Let us apply the Cauchy--Schwarz inequality on the right hand side of the above inequality, we get 
\begin{equation}
     \left|\frac{{\rm d} }{{\rm d}t}Q(\cal{A}_0,\cal{A}_t)\right| \leq 2\sqrt{2}\norm{[\cal{A}_0,\cal{L}^{\dagger}({\cal{A}}_t)]}_{\rm HS}\sqrt{Q(\cal{A}_0,\cal{A}_t)}.
\end{equation}
From above inequality we obtain
\begin{equation}
     \left|\int_0^{T}\frac{{{\rm d} Q(\cal{A}_0,\cal{A}_t)} }{\sqrt{Q(\cal{A}_0,\cal{A}_t)}}\right| \leq 2\sqrt{2}\int_{0}^{T}\norm{[\cal{A}_0,\cal{L}^{\dagger}({\cal{A}}_t)]}_{\rm HS}{\rm d}t.
\end{equation}

After integrating above inequality, we obtain the following bound 
\begin{equation}
  T \geq T_{Q} = \frac{\sqrt{Q(\mathcal{A}_0, \mathcal{A}_T)}}{\sqrt{2}\langle\!\langle \lVert [\mathcal{A}_0, \mathcal{L}^{\dagger}(\mathcal{A}_t)] \rVert_{\rm HS} \rangle\!\rangle}_{\rm T},
\end{equation}
\end{proof}

The above theorem provides an interesting inequality between generation of quantumness and evolution time. The saturation of above bound can depend on both observable and type of dynamics. For given observable the dynamics will saturate above bound will take minimum amount of time to generate certain amount of quantumness.

In above theorem if we replace observable by state and if CPTP dynamical map is self adjoint then we can recover the speed limits on quantumness of states which is derived in Refs~\cite{Jing2016}. Therefore speed limits on quantumness of states can be considered as special case of speed limits on quantumness of observables.

\begin{corollary}
For a given observable $\mathcal{A}_0$ of a quantum system whose time evolution is governed by the adjoint of the Liouvillian super-operator $\mathcal{L}^{\dagger}_t$, the time required to generate a certain amount of quantumness with respect to a reference observable $\mathcal{B}_0$, which is expressed by $Q(\mathcal{B}_0, \mathcal{A}_T)$ is lower bounded by
\begin{equation}
  T \geq T_{Q} = \frac{\left|\sqrt{Q(\mathcal{B}_0, \mathcal{A}_0)}-\sqrt{Q(\mathcal{B}_{0} \mathcal{A}_T)}\right|}{\sqrt{2}\langle\!\langle \lVert [\mathcal{B}_0, \mathcal{L}^{\dagger}(\mathcal{A}_t)] \rVert_{\rm HS} \rangle\!\rangle_{\rm T}},
\end{equation}
where $\mathcal{A}_0$ and $\mathcal{A}_T$ represent the initial and final observables of the given quantum system, respectively. Additionally, $\langle\!\langle X_t \rangle\!\rangle_{T} \equiv \frac{1}{T} \int_{0}^{T} X_t \,dt$ denotes the time average of the quantity $X_t$.
\end{corollary}
The proof can be done by replacing $\mathcal{A}_0$ with $\mathcal{B}_0$ in the proof of Theorem \ref{Th1}.\\

We now derive a speed limit on the Skew information which is an informational measure of quantumness. The detailed proof for the following result is given in Appendix~\ref{Appendix:proof_cor_2}.

\begin{corollary}\label{Th2}
For a d-dimensional quantum system which is described by $\rho$, the minimum time required for the quantum system to attain skew information $\cal{I}(\rho,\cal{A}_T)$ from initial skew information $\cal{I}(\rho,\cal{A}_0)$ through arbitrary dynamics is lower bounded by

\begin{equation}
T\geq T_{Q}= \frac{\sqrt{2}|\sqrt{I(\rho,\cal{A}_T)}-\sqrt{I(\rho,\cal{A}_0)}|}{\langle\!\langle \lVert [\sqrt{\mathcal{\rho}}, \mathcal{L}^{\dagger}(\mathcal{A}_t)] \rVert_{\rm HS} \rangle\!\rangle_{\rm T}},
\end{equation}
where $\mathcal{A}_0$ and $\mathcal{A}_T$ represent the initial and final observables of the given quantum system, respectively. Additionally, $\langle\!\langle X_t \rangle\!\rangle_{T} \equiv \frac{1}{T} \int_{0}^{T} X_t \,dt$ denotes the time average of the quantity $X_t$.
\end{corollary}

This can be proved by replacing $\mathcal{A}_0 \rightarrow \sqrt{\rho}$ in the proof of Theorem \ref{Th1}.

Let us now see how this result can be applied to infer a speed limit on quantum coherence. Recently in \cite{Yu2017} it was shown that Skew information of a state $\rho$ can be used to define its coherence with respect to an eigenbasis $\{\ketbra{k}\}_k$ of the observable $\mathcal{A}$. Quantum coherence is a useful resource in quantum information theory, quantum computation, quantum thermodynamics, and the developing quantum technologies. Due to the presence of unwanted environmental interactions and noise (such as decoherence and dissipation), quantum systems lose their coherence, affecting the performance of quantum devices. Thus, one of the most significant challenges in quantum technologies is maintaining coherence in quantum states for longer periods. Another perspective is understanding the physical process under which a certain amount of coherence can be created in the quantum system while expending minimal amount of time. Consequently, determining the minimum timescale over which a certain amount of quantum coherence is generated or degraded in a physical process is important for quantum technologies. To address this question, we establish speed limits on quantum coherence. The speed limits on quantum coherence represent upper bounds on the instantaneous rate of change of quantum coherence and provide the minimum timescale required to bring about specific changes in quantum coherence under arbitrary dynamics. Since coherence is defined with respect to states, we will use the Schrodinger picture in the following.

\begin{theorem}\label{Th3}
For a d-dimensional quantum system, the minimum time required for the quantum system to attain coherence $\cal{C}(\rho_T,\cal{A})$ from initial coherence $\cal{C}(\rho_0,\cal{A})$ through arbitrary dynamics is lower bounded by

\begin{equation}
T\geq T_{\mathcal{C}}= \frac{\sqrt{2}|\sqrt{\cal{C}(\rho_0,\cal{A})}-\sqrt{\cal{C}(\rho_T,\cal{A})}|}{\langle\!\langle\sqrt{\sum_{k}{ \norm{[\partial_t\sqrt{\rho_t},{\ketbra{k}}]}^2_{\rm HS}}} \rangle\!\rangle_{\rm T}}.
\end{equation}

\noindent where $\rho_0$ and $\rho_T$ represent the initial and final states of the given quantum system, respectively. Here, coherence of given quantum system is measured in the basis of an arbitrary observable $\cal{A}$, i.e. $\{\ketbra{k}\}_{0}^{d-1}$. Additionally, $\langle\!\langle X_t \rangle\!\rangle_{T} \equiv \frac{1}{T} \int_{0}^{T} X_t \,dt$, denotes the time average of the quantity $X_t$.
\end{theorem}

\begin{proof}
    
The coherence at time $t$ can be expressed as
\begin{equation}
\cal{C}(\rho_{t},\cal{A}) = \frac{1}{2}\sum_{k}\norm{[\sqrt{\rho_{t}}, \ketbra{k}]}_{\rm HS}^{2}.
\end{equation}


After differentiating the above equation with respect to time $t$, we obtain
\begin{align}
\frac{{\rm d} }{{\rm d}t}\cal{C}(\rho_t,\cal{A}) = &\frac{1}{2}\sum_{k}\Big(\tr\big([ \partial_{t}\sqrt{\rho_{t}},\ketbra{k}]^{\dagger}[\sqrt{\rho_t},\ketbra{k}]\big)\nonumber\\
&+\tr\big([\sqrt{\rho_{t}},\ketbra{k}]^{\dagger}[\partial_{t}\sqrt{\rho_{t}},{\ketbra{k}}]\big)\Big).
\end{align}
Let us now consider the absolute value of the above equation and apply triangular inequality $|A+B|\leq |A|+|B|$ and Cauchy--Schwarz inequality on the right hand side of the above inequality. We then obtain the following inequality
\begin{equation}
     \left|\frac{{\rm d} }{{\rm d}t}\cal{C}(\rho_t,A)\right| \leq \sum_{k}\norm{[\partial_t \sqrt{\rho_t},\ketbra{k}]}_{\rm HS}\norm{[\sqrt{\rho_t},\ketbra{k}]}_{\rm HS}.
\end{equation}

If we again apply the Cauchy-Schwarz inequality of sum ( $(\sum_{i=0}^{n} a_i b_i)^2 \leq \sum_{i=0}^{n}  a_{i}^{2} \sum_{i=0}^{n}  b_{i}^{2} $ , where $a_i^2$'s and $b_i^2$'s are real numbers) in right hand side of above inequality , we obtain

\begin{equation}\label{cohspeed}
     \left|\frac{{\rm d} }{{\rm d}t}\cal{C}(\rho_t,A)\right| \leq \sqrt{2}\sqrt{\sum_{k}\norm{[\partial_t \sqrt{\rho_t},\ketbra{k}]}^{2}_{\rm HS}}\sqrt{\cal{C}(\rho_t, A)}.
\end{equation}

From above inequality we obtain
\begin{equation}
     \left|\int_0^{T}\frac{{{\rm d} \cal{C}(\rho_t,\cal{A})} }{\sqrt{\cal{C}(\rho_t,\cal{A})}}\right| \leq \sqrt{2}\int_{0}^{T}\sqrt{\sum_{k}\norm{[\partial \sqrt{\rho_t},{\ketbra{k}}]}^2_{\rm HS}}{\rm d}t.
\end{equation}

After integrating above inequality, we obtain the following bound 
\begin{equation}\label{cohtime}
T\geq T_{\mathcal{C}}= \frac{\sqrt{2}|\sqrt{\cal{C}(\rho_0,\cal{A})}-\sqrt{\cal{C}(\rho_T,\cal{A})}|}{\langle\!\langle\sqrt{\sum_{k}{ \norm{[\partial_t\sqrt{\rho_t},{\ketbra{k}}]}^2_{\rm HS}}}\rangle\!\rangle_{\rm T}}.
\end{equation}
\end{proof}
The inequality~\eqref{cohspeed} represents the upper bound on the instantaneous rate of change of the coherence of the quantum system for arbitrary quantum dynamics. From the bound we see that this rate depends on instantaneous evolution speed of coherence and the coherence of the system in that instant. The speed limit on coherence applies to both coherence production and degradation processes. It can characterise the ability of coherence generation and degradation by arbitrary quantum dynamics in time. For completely dephasing process, the above bound~\eqref{cohtime} is interpreted as the speed limit on decoherence, which determines minimum timescale over which a coherent state decoheres with the environment. The bound~\eqref{cohtime} saturate for those dynamics which generate or degrade coherence in fastest possible way. It is important to note that the above speed limit on coherence is arguably easier to calculate than previously derived speed limits on coherence using relative entropy of coherence~\cite{Mohan2022}. Now we would like to demonstrate the performance and tightness of the discussed bounds. It is imperative to mention that there have been some earlier works on speed limit bounds on open quantum system for state evolution~\cite{Campo2013,Mirkin-2016}, here in the next section we would like to illustrate using the obtained bounds in this article for quantumness through observable evolution as well as coherence generation.

\section{Illustrative Example}\label{sec:Illustrations}
In this section, we will estimate the speed limits derived in previous section for dephasing, pure-dephasing and unitary processes and demonstrate that the lower bound is tight in these cases.

The dephasing channel is significant in quantum information because it models how environmental noise affects the coherence of quantum states by reducing the off-diagonal elements of a system’s density matrix. This loss of coherence is critical because quantum algorithms and protocols, such as quantum computation and communication, rely heavily on maintaining superposition and entanglement. 
Similalry, the pure dephasing channel is particularly interesting in quantum information because it represents a type of noise that affects the phase coherence of quantum states without altering their energy populations. Unlike other types of decoherence, which might involve energy dissipation (such as amplitude damping), pure dephasing exclusively impacts the off-diagonal elements of a quantum system’s density matrix in the computational basis. This makes it a crucial model for understanding how quantum coherence, a vital resource for quantum computation and communication, is degraded in real-world quantum systems. Since many quantum algorithms and protocols rely on superposition and entanglement, which are inherently sensitive to phase coherence, studying the effects of dephasing and pure dephasing can help in developing error-correction techniques and designing systems that can mitigate the loss of quantum information. 

\subsection{Quantumness Generation}

{\it  General dephasing process--} To begin with, we look at bound for quantumness generated under arbitrary evolution. 
Since quantumness is defined in terms of observables, we will evaluate the relevant lower bounds in Heisenberg picture as follows: 
\begin{eqnarray}\label{heismastereqn}
    \frac{{\rm d}\mathcal{A}_t}{{\rm dt}}= i [H,\mathcal{A}_t]+ \frac{\gamma}{2}(\sigma_z \mathcal{A}_t \sigma_z - \mathcal{A}_t)
\end{eqnarray}
Here, $\gamma$ is time independent if the evolution is Markovian and time dependent for non-Markovian evolution. The Markovian evolution is unitary when $\gamma=0$ and dephasing when $\gamma$ is non-zero. We have given an alternate analytical solution to the case for quantumness bound under unitary evolution in Appendix~\ref{Appendix:unitary}.

For our example, we consider the cases $\gamma=0.01$ and $\gamma=0$ (unitary) along with $\mathcal{A}_0 = \sigma_y, \ H=\sigma_x$. The non-zero coefficient $\gamma$ cannot be arbitrary but needs to be small to continue under markovian dynamics.  We compute the observables at $t$ by vectorizing $\mathcal{A}_0$ and the master Eq. \eqref{heismastereqn}, as given in the Appendix~\ref{Appendix:dephasing}.

{\textit{Case 1} ($\gamma=0.01$)}-- The observable $\mathcal{A}_t$ given by Eq. \eqref{At} for our example is

\begin{align}
    &\mathcal{A}_t =\nonumber\\
    &e^{-0.005t} \begin{bmatrix}
        -\sin{2t} & i(0.0025 \sin{2t}- \cos{2t} \\
        -i(0.0025 \sin{2t}- \cos{2t}) & \sin{2t}
    \end{bmatrix}.
\end{align}
The corresponding value of Quantumness is given by

\begin{eqnarray}
    Q(\mathcal{A}_0,\mathcal{A}_t) = 16 e^{-0.01 t}\sin^2{2t}.   
\end{eqnarray}
The integrand in the denominator of $T_{QSL}$ is:
\begin{eqnarray}
    ||[\mathcal{A}_0,\mathcal{L}^{\dag}(\mathcal{A}_t)]||_{HS} = 2\sqrt{2} e^{-0.01t}|0.005\sin{2t}- 2\cos{2t}|. \nonumber\\
\end{eqnarray}

After performing the integration numerically the resulting $T_{QSL}$ is given by the green dotted line in Fig.~\ref{qness}.
\\

{\textit{Case 2} ($\gamma=0$)}-- The detailed analytic evaluation of $T_{QSL}$ for a general observable and Hamiltonian is done in Appendix. Here, we state the results for $\mathcal{A}_0 = \sigma_y, \ H=\sigma_x$:

\begin{eqnarray}
    \mathcal{A}_t =
    \begin{bmatrix}
        -\sin{2t} & -i\cos{2t} \\
        i\cos{2t} & \sin{2t}
    \end{bmatrix} .
\end{eqnarray}

The corresponding values of quantumness is given by:

\begin{eqnarray}
     Q(\mathcal{A}_0,\mathcal{A}_t) = 16 \sin^2 {2t}.
\end{eqnarray}

The integrand in the denominator of $T_{QSL}$ is given by
\begin{eqnarray}
    ||[\mathcal{A}_0,\mathcal{L}^{\dag}(\mathcal{A}_t)]||_{HS} = 4\sqrt{2}|\cos{2t}|.
\end{eqnarray}

These functions are numerically integrated to obtain the Plots. Fig.~\ref{fig:qness} shows that the QSL bound for quantumness is indeed tight.\\

\begin{figure}[ht]
    \centering
    \includegraphics[width=8cm]{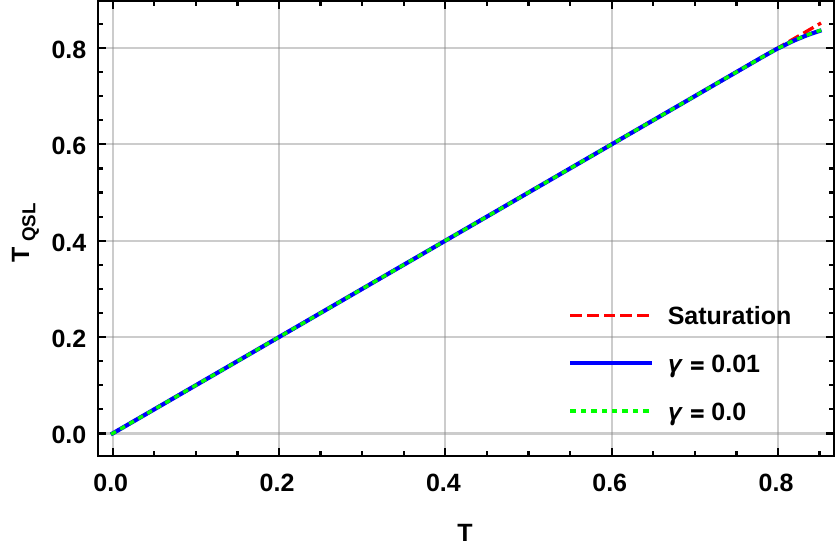}
    \caption{Here is the plot for \(T_{QSL}\) vs $T$ $\in[0,0.85]$ for generation of Quantumness under Open dynamics.}
    \label{fig:qness}
\end{figure}

{\it \bf Pure dephasing process:} We illustrate this bound for two examples. First, namely dephasing model, which is spin-boson interaction between qubit system and a bosonic reservoir. The dynamics of an observable $A_{t}$ is described by the master equation
\begin{equation}
    \mathcal{L}_{t}(A_{t})=\frac{\gamma_{t}}{2}\left(\sigma_{z}A_{t}\sigma_{z}-A_{t}\right).
\end{equation}
A general qubit observable at time $t$ has the following form as given in~\cite{Wu-20},
\begin{equation}
    A_{t}= \begin{pmatrix}
    a_3  && (a_1 -i a_2)e^{-g(t)} \\
    (a_1 + i a_2 )e^{-g(t)} && -a_3 
    \end{pmatrix},
\end{equation}
where \(g(t) = \int_{0}^{t}dt' \gamma_{t'}\) is the dephasing factor. For this example, the dephasing rate $\gamma_{t}$, has analytical form \cite{Wu-20,Chin2012} \(\gamma_{t}=\eta(1+t^{2})^{-s/2}\Gamma(s)\sin[s\arctan{t}]\),
with $\Gamma(\cdot)$ being Euler gamma function and $\eta$ a dimensionless constant. The property of environment is determined by parameter $s$ which divides the reservoir into sub-Ohmic $(s<1)$, Ohmic $(s=1)$ and super-Ohmic $(s>1)$ reservoirs.

Evaluating bound for this case requires following functions which can be algebraically obtained
\begin{align}
    Q(A_{0},A_{t}) &= 2\tr\left([A_{0},A_{t}]^{\dagger}[A_{0},A_{t}]\right) \nonumber\\
    &= 16 e^{-2 g(t)}(e^{g(t)}-1)^{2}a_{3}^{2}(a_{1}^{2}+a_{2}^{2}).
\end{align}
\begin{equation}
    \mathcal{L}^{\dag}(A_{t})= -\gamma_{t}e^{-g(t)}\begin{bmatrix}
    0  && a_1 -i a_2 \\
    a_1 + i a_2  && 0 
    \end{bmatrix}.
\end{equation}


\begin{equation}
    \|[A_{0},\mathcal{L}^{\dagger} (A_{t})]\|_{HS} = 2\sqrt{2}e^{-g(t)}|\gamma_{t}||a_{3}|\sqrt{a_{1}^{2}+a_{2}^{2}}.
\end{equation}    
Plugging these terms into Eq~\eqref{QQSL}, we obtain
\begin{equation}
    T\geq T_{Q} = \frac{1-e^{-g(T)}}{\frac{1}{T}\int_{0}^{T}dt |\gamma_{t}|e^{-g(t)}}. 
\end{equation}
As illustrated in Fig~\ref{fig:coherence_gen}, the QSl for quantum coherence is also tight.

\begin{figure}[ht]
    \centering
    \includegraphics[width=8cm]{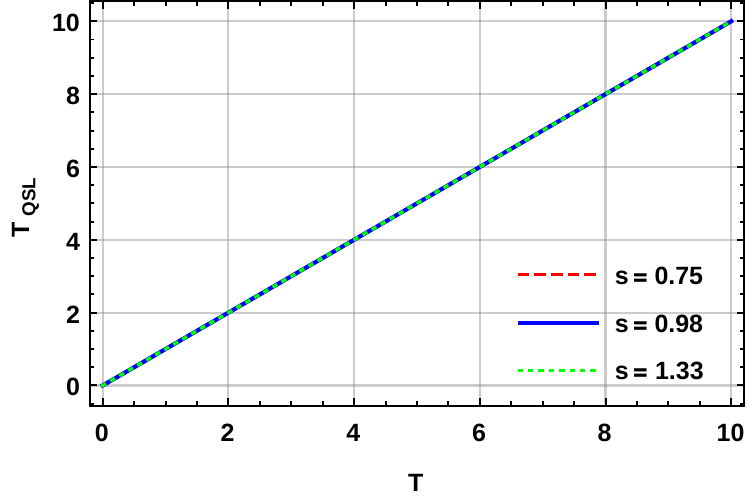}
    \caption{Here is the plot for \(T_{QSL}\) vs $T$ $\in[0,1.0]$ for generation of Quantumness under Non-Markovian dynamics with different values of $s$.}
    \label{fig:speed_limit_for_Non-Markovian}
\end{figure}

\subsection{Coherence generation}
Quantum coherence is essential in quantum computing because it enables the superposition of quantum states, a fundamental resource that allows quantum computers to perform complex calculations more efficiently than classical computers.  Many quantum algorithms rely on the coherent manipulation of qubits to achieve speedups over classical counterparts. Maintaining coherence is also critical for entanglement, another resource that enhances quantum processing power. Without coherence, quantum systems lose their advantage, making its preservation vital for the success of quantum computing.

In this case, we are working in the Schrodinger picture with $\mathcal{A}=\sigma_z$ (which does not evolve) and $\rho_0 = \ketbra{0}$ (which evolves via $H= \sigma_x$). This evolution in the Schrodinger picture is given by the Lindblad master equation:

\begin{eqnarray}\label{schrodmastereqn}
    \frac{{\rm d}\rho_t}{{\rm dt}}= -i [H,\rho_t]+ \frac{\gamma}{2}(\sigma_z \rho_t \sigma_z - \rho_t)
\end{eqnarray}

Using the vectorization techniques introduced in the Appendix, we get the following state at time $t$:

\begin{align}
    &\rho_t =\frac{\openone}{2}+ \nonumber\\    
    & e^{-0.005 t}\begin{bmatrix}
        \left( \frac{1}{2} \cos{2t} +0.0012 \sin{2t} \right) & \frac{i}{2} \sin{2t} \\
        -\frac{i}{2} \sin{2t} & \left( \frac{1}{2} \cos{2t} +0.0012 \sin{2t} \right)
    \end{bmatrix}
\end{align}

At $t=0$, $C(\rho_0, \mathcal{A})=0$ as $\rho_0$ is an eigenstate of $\mathcal{A}$. At time $t$, the coherence is given by

\begin{eqnarray}
    C(\rho_t,\mathcal{A})= |\sin{2t}|e^{-0.005t}
\end{eqnarray}

The integrand of the denominator is given by

\begin{align}
    \sum_k &||[\partial_t \sqrt{\rho_t}, \ketbra{k}]||^2 =\\
    &\frac{2 e^{-0.005t}(0.5+0.5 \cos{4t}-0.0025\sin{4t})}{|\sin{2t}|}.
\end{align}
After integrating numerically, we obtain the plots in Fig. \ref{fig:coherence_gen}. As we can see from the Fig, the bound is indeed tight.

\vspace{0.4cm}
\begin{figure}
    \centering
    \includegraphics[width=8cm]{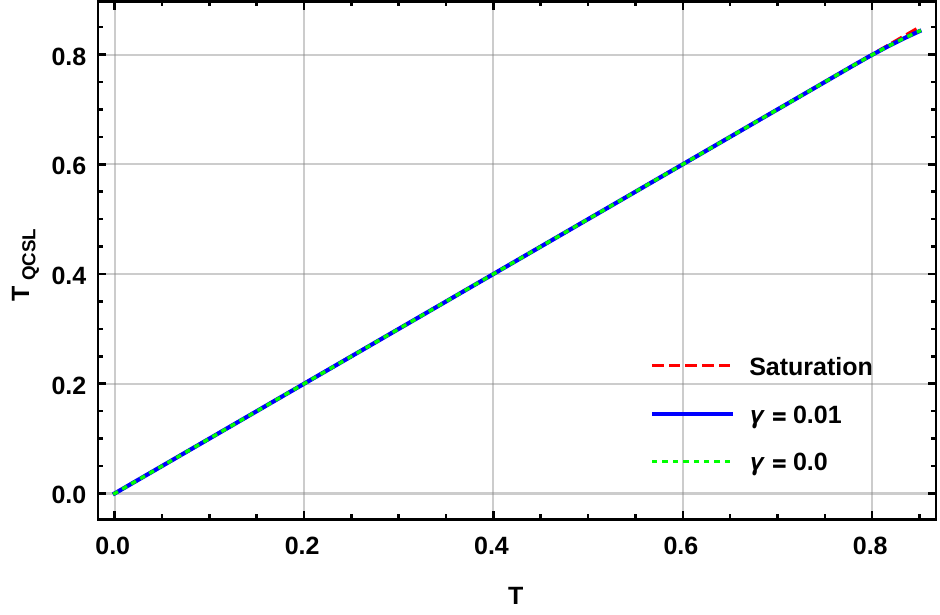}
    \caption{Here is the plot for \(T_{CQSL}\) vs $T$ $\in[0,0.85]$ for generation of Quantum Coherence under Non-Markovian dynamics with $\gamma=0.01$ and Unitary dynamics with $\gamma=0$.}
    \label{fig:coherence_gen}
\end{figure}


\section{Conclusion}\label{sec:Conclusion}
In this work, we have derived speed limits on various measures of quantumness, such as the non-commutativity of observables, skew information, and quantum coherence. These speed limits provide the minimum timescale required to generate and degrade quantumness in a physical process. The speed limits on quantumness have been computed and analyzed in a variety of relevant physical processes, including unitary evolution and pure dephasing dynamics. Moreover, we have found that the speed limits on quantumness saturate for pure dephasing dynamics. It is important to note that the speed limits derived in this paper are fundamentally different from standard speed limits, which are based on the distinguishability of the initial and final states of the given system. The speed limits on quantumness reveals the coherence generating or degrading ability of quantum dynamics in a given amount of time. 
We have illustrated QSL for quantumness and quantum coherence for dephasing as well as pure dephasing channels.
Dephasing channel plays a significant role in quantum information because it models how environmental noise affects the coherence of quantum states by reducing the off-diagonal elements of a system’s density matrix. This loss of coherence is critical because quantum algorithms and protocols, such as quantum computation and communication, rely heavily on maintaining superposition and entanglement. Understanding how long quantumness and quantum coherence can be maintained may help in the development of strategies to counteract this noise, including error correction and fault-tolerant quantum computing, which are essential for building reliable quantum technologies. 
We believe that speed limits on quantumness will have applications in developing quantum technologies such as quantum computing, quantum communication, quantum control, and quantum energetic and thermal devices. In future, our results can be applied to certain solid-state quantum systems, like superconducting qubits and quantum dots, making it highly relevant for practical implementations of quantum technologies.

\vskip .5cm
\noindent
{\it Acknowledgments:}
We are very thankful to Brij Mohan for making useful remarks and suggestions during the working out of the obtained bounds. D.S. and S.B. acknowledge the support of the INFOSYS scholarship during the initial phase of the work. D.S. acknowledges TCG CREST (CQuERE) for their support and hospitality during the two week visit.

\bibliography{name.bib}

\begin{thebibliography}{72}%
\makeatletter
\providecommand \@ifxundefined [1]{%
 \@ifx{#1\undefined}
}%
\providecommand \@ifnum [1]{%
 \ifnum #1\expandafter \@firstoftwo
 \else \expandafter \@secondoftwo
 \fi
}%
\providecommand \@ifx [1]{%
 \ifx #1\expandafter \@firstoftwo
 \else \expandafter \@secondoftwo
 \fi
}%
\providecommand \natexlab [1]{#1}%
\providecommand \enquote  [1]{``#1''}%
\providecommand \bibnamefont  [1]{#1}%
\providecommand \bibfnamefont [1]{#1}%
\providecommand \citenamefont [1]{#1}%
\providecommand \href@noop [0]{\@secondoftwo}%
\providecommand \href [0]{\begingroup \@sanitize@url \@href}%
\providecommand \@href[1]{\@@startlink{#1}\@@href}%
\providecommand \@@href[1]{\endgroup#1\@@endlink}%
\providecommand \@sanitize@url [0]{\catcode `\\12\catcode `\$12\catcode `\&12\catcode `\#12\catcode `\^12\catcode `\_12\catcode `\%12\relax}%
\providecommand \@@startlink[1]{}%
\providecommand \@@endlink[0]{}%
\providecommand \url  [0]{\begingroup\@sanitize@url \@url }%
\providecommand \@url [1]{\endgroup\@href {#1}{\urlprefix }}%
\providecommand \urlprefix  [0]{URL }%
\providecommand \Eprint [0]{\href }%
\providecommand \doibase [0]{http://dx.doi.org/}%
\providecommand \selectlanguage [0]{\@gobble}%
\providecommand \bibinfo  [0]{\@secondoftwo}%
\providecommand \bibfield  [0]{\@secondoftwo}%
\providecommand \translation [1]{[#1]}%
\providecommand \BibitemOpen [0]{}%
\providecommand \bibitemStop [0]{}%
\providecommand \bibitemNoStop [0]{.\EOS\space}%
\providecommand \EOS [0]{\spacefactor3000\relax}%
\providecommand \BibitemShut  [1]{\csname bibitem#1\endcsname}%
\let\auto@bib@innerbib\@empty
\bibitem [{\citenamefont {Heisenberg}(1927)}]{Heisenberg}%
  \BibitemOpen
  \bibfield  {author} {\bibinfo {author} {\bibfnamefont {W.}~\bibnamefont {Heisenberg}},\ }\bibfield  {title} {\enquote {\bibinfo {title} {Über den anschaulichen inhalt der quantentheoretischen kinematik und mechanik},}\ }\href {\doibase 10.1007/BF01397280} {\bibfield  {journal} {\bibinfo  {journal} {Zeitschrift für Physik}\ }\textbf {\bibinfo {volume} {43}},\ \bibinfo {pages} {172--198} (\bibinfo {year} {1927})}\BibitemShut {NoStop}%
\bibitem [{\citenamefont {Robertson}(1929)}]{Robertson1929}%
  \BibitemOpen
  \bibfield  {author} {\bibinfo {author} {\bibfnamefont {H.~P.}\ \bibnamefont {Robertson}},\ }\bibfield  {title} {\enquote {\bibinfo {title} {The uncertainty principle},}\ }\href {\doibase 10.1103/PhysRev.34.163} {\bibfield  {journal} {\bibinfo  {journal} {Physical Review}\ }\textbf {\bibinfo {volume} {34}},\ \bibinfo {pages} {163--164} (\bibinfo {year} {1929})}\BibitemShut {NoStop}%
\bibitem [{\citenamefont {Pati}\ and\ \citenamefont {Sahu}(2007)}]{Pati-2007}%
  \BibitemOpen
  \bibfield  {author} {\bibinfo {author} {\bibfnamefont {A.K.}\ \bibnamefont {Pati}}\ and\ \bibinfo {author} {\bibfnamefont {P.K.}\ \bibnamefont {Sahu}},\ }\bibfield  {title} {\enquote {\bibinfo {title} {Sum uncertainty relation in quantum theory},}\ }\href {\doibase 10.1016/j.physleta.2007.03.005} {\bibfield  {journal} {\bibinfo  {journal} {Physics Letters A}\ }\textbf {\bibinfo {volume} {367}},\ \bibinfo {pages} {177–181} (\bibinfo {year} {2007})}\BibitemShut {NoStop}%
\bibitem [{\citenamefont {Chen}\ and\ \citenamefont {Fei}(2015)}]{Chen-2015}%
  \BibitemOpen
  \bibfield  {author} {\bibinfo {author} {\bibfnamefont {Bin}\ \bibnamefont {Chen}}\ and\ \bibinfo {author} {\bibfnamefont {Shao-Ming}\ \bibnamefont {Fei}},\ }\bibfield  {title} {\enquote {\bibinfo {title} {Sum uncertainty relations for arbitrary n incompatible observables},}\ }\href {\doibase 10.1038/srep14238} {\bibfield  {journal} {\bibinfo  {journal} {Scientific Reports}\ }\textbf {\bibinfo {volume} {5}},\ \bibinfo {pages} {14238} (\bibinfo {year} {2015})}\BibitemShut {NoStop}%
\bibitem [{\citenamefont {Chen}\ \emph {et~al.}(2016)\citenamefont {Chen}, \citenamefont {Fei},\ and\ \citenamefont {Long}}]{Chen-2016}%
  \BibitemOpen
  \bibfield  {author} {\bibinfo {author} {\bibfnamefont {Bin}\ \bibnamefont {Chen}}, \bibinfo {author} {\bibfnamefont {Shao-Ming}\ \bibnamefont {Fei}}, \ and\ \bibinfo {author} {\bibfnamefont {Gui-Lu}\ \bibnamefont {Long}},\ }\bibfield  {title} {\enquote {\bibinfo {title} {Sum uncertainty relations based on wigner–yanase skew information},}\ }\href {\doibase 10.1007/s11128-016-1274-3} {\bibfield  {journal} {\bibinfo  {journal} {Quantum Information Processing}\ }\textbf {\bibinfo {volume} {15}},\ \bibinfo {pages} {2639–2648} (\bibinfo {year} {2016})}\BibitemShut {NoStop}%
\bibitem [{\citenamefont {Maccone}\ and\ \citenamefont {Pati}(2014)}]{Pati-2014}%
  \BibitemOpen
  \bibfield  {author} {\bibinfo {author} {\bibfnamefont {Lorenzo}\ \bibnamefont {Maccone}}\ and\ \bibinfo {author} {\bibfnamefont {Arun~K.}\ \bibnamefont {Pati}},\ }\bibfield  {title} {\enquote {\bibinfo {title} {Stronger uncertainty relations for all incompatible observables},}\ }\href {\doibase 10.1103/PhysRevLett.113.260401} {\bibfield  {journal} {\bibinfo  {journal} {Phys. Rev. Lett.}\ }\textbf {\bibinfo {volume} {113}},\ \bibinfo {pages} {260401} (\bibinfo {year} {2014})}\BibitemShut {NoStop}%
\bibitem [{\citenamefont {Englert}(2024)}]{Englert-2024}%
  \BibitemOpen
  \bibfield  {author} {\bibinfo {author} {\bibfnamefont {Berthold-Georg}\ \bibnamefont {Englert}},\ }\bibfield  {title} {\enquote {\bibinfo {title} {Uncertainty relations revisited},}\ }\href {\doibase https://doi.org/10.1016/j.physleta.2023.129278} {\bibfield  {journal} {\bibinfo  {journal} {Physics Letters A}\ }\textbf {\bibinfo {volume} {494}},\ \bibinfo {pages} {129278} (\bibinfo {year} {2024})}\BibitemShut {NoStop}%
\bibitem [{\citenamefont {Mondal}\ \emph {et~al.}(2017)\citenamefont {Mondal}, \citenamefont {Bagchi},\ and\ \citenamefont {Pati}}]{Mondal-2017}%
  \BibitemOpen
  \bibfield  {author} {\bibinfo {author} {\bibfnamefont {Debasis}\ \bibnamefont {Mondal}}, \bibinfo {author} {\bibfnamefont {Shrobona}\ \bibnamefont {Bagchi}}, \ and\ \bibinfo {author} {\bibfnamefont {Arun~Kumar}\ \bibnamefont {Pati}},\ }\bibfield  {title} {\enquote {\bibinfo {title} {Tighter uncertainty and reverse uncertainty relations},}\ }\href {\doibase 10.1103/PhysRevA.95.052117} {\bibfield  {journal} {\bibinfo  {journal} {Phys. Rev. A}\ }\textbf {\bibinfo {volume} {95}},\ \bibinfo {pages} {052117} (\bibinfo {year} {2017})}\BibitemShut {NoStop}%
\bibitem [{\citenamefont {Shrimali}\ \emph {et~al.}(2022)\citenamefont {Shrimali}, \citenamefont {Bhowmick}, \citenamefont {Pandey},\ and\ \citenamefont {Pati}}]{Shrimali2022}%
  \BibitemOpen
  \bibfield  {author} {\bibinfo {author} {\bibfnamefont {Divyansh}\ \bibnamefont {Shrimali}}, \bibinfo {author} {\bibfnamefont {Swapnil}\ \bibnamefont {Bhowmick}}, \bibinfo {author} {\bibfnamefont {Vivek}\ \bibnamefont {Pandey}}, \ and\ \bibinfo {author} {\bibfnamefont {Arun~Kumar}\ \bibnamefont {Pati}},\ }\bibfield  {title} {\enquote {\bibinfo {title} {Capacity of entanglement for a nonlocal hamiltonian},}\ }\href {\doibase 10.1103/PhysRevA.106.042419} {\bibfield  {journal} {\bibinfo  {journal} {Phys. Rev. A}\ }\textbf {\bibinfo {volume} {106}},\ \bibinfo {pages} {042419} (\bibinfo {year} {2022})}\BibitemShut {NoStop}%
\bibitem [{\citenamefont {Hasegawa}(2020)}]{Hasegawa2020}%
  \BibitemOpen
  \bibfield  {author} {\bibinfo {author} {\bibfnamefont {Yoshihiko}\ \bibnamefont {Hasegawa}},\ }\bibfield  {title} {\enquote {\bibinfo {title} {Quantum thermodynamic uncertainty relation for continuous measurement},}\ }\href {\doibase 10.1103/PhysRevLett.125.050601} {\bibfield  {journal} {\bibinfo  {journal} {Phys. Rev. Lett.}\ }\textbf {\bibinfo {volume} {125}},\ \bibinfo {pages} {050601} (\bibinfo {year} {2020})}\BibitemShut {NoStop}%
\bibitem [{\citenamefont {Yu}\ \emph {et~al.}(2022)\citenamefont {Yu}, \citenamefont {Liu}, \citenamefont {Yang}, \citenamefont {Gong}, \citenamefont {Cao}, \citenamefont {Zhang}, \citenamefont {Liu}, \citenamefont {Heyl}, \citenamefont {Ozawa}, \citenamefont {Goldman},\ and\ \citenamefont {Cai}}]{Yu2022}%
  \BibitemOpen
  \bibfield  {author} {\bibinfo {author} {\bibfnamefont {Min}\ \bibnamefont {Yu}}, \bibinfo {author} {\bibfnamefont {Yu}~\bibnamefont {Liu}}, \bibinfo {author} {\bibfnamefont {Pengcheng}\ \bibnamefont {Yang}}, \bibinfo {author} {\bibfnamefont {Musang}\ \bibnamefont {Gong}}, \bibinfo {author} {\bibfnamefont {Qingyun}\ \bibnamefont {Cao}}, \bibinfo {author} {\bibfnamefont {Shaoliang}\ \bibnamefont {Zhang}}, \bibinfo {author} {\bibfnamefont {Haibin}\ \bibnamefont {Liu}}, \bibinfo {author} {\bibfnamefont {Markus}\ \bibnamefont {Heyl}}, \bibinfo {author} {\bibfnamefont {Tomoki}\ \bibnamefont {Ozawa}}, \bibinfo {author} {\bibfnamefont {Nathan}\ \bibnamefont {Goldman}}, \ and\ \bibinfo {author} {\bibfnamefont {Jianming}\ \bibnamefont {Cai}},\ }\bibfield  {title} {\enquote {\bibinfo {title} {Quantum fisher information measurement and verification of the quantum cram{\'e}r--rao bound in a solid-state qubit},}\ }\href {\doibase 10.1038/s41534-022-00547-x} {\bibfield  {journal} {\bibinfo  {journal} {npj Quantum
  Information}\ }\textbf {\bibinfo {volume} {8}},\ \bibinfo {pages} {56} (\bibinfo {year} {2022})}\BibitemShut {NoStop}%
\bibitem [{\citenamefont {Ray}\ \emph {et~al.}(2022)\citenamefont {Ray}, \citenamefont {Ghoshal}, \citenamefont {Pati},\ and\ \citenamefont {Sen}}]{Tanaya-22}%
  \BibitemOpen
  \bibfield  {author} {\bibinfo {author} {\bibfnamefont {Tanaya}\ \bibnamefont {Ray}}, \bibinfo {author} {\bibfnamefont {Ahana}\ \bibnamefont {Ghoshal}}, \bibinfo {author} {\bibfnamefont {Arun~Kumar}\ \bibnamefont {Pati}}, \ and\ \bibinfo {author} {\bibfnamefont {Ujjwal}\ \bibnamefont {Sen}},\ }\bibfield  {title} {\enquote {\bibinfo {title} {Estimating quantum coherence by noncommutativity of any observable and its incoherent part},}\ }\href {\doibase 10.1103/PhysRevA.105.062423} {\bibfield  {journal} {\bibinfo  {journal} {Phys. Rev. A}\ }\textbf {\bibinfo {volume} {105}},\ \bibinfo {pages} {062423} (\bibinfo {year} {2022})}\BibitemShut {NoStop}%
\bibitem [{\citenamefont {Mandelstam}\ and\ \citenamefont {Tamm}(1945)}]{Mandelstam1945}%
  \BibitemOpen
  \bibfield  {author} {\bibinfo {author} {\bibfnamefont {Leonid}\ \bibnamefont {Mandelstam}}\ and\ \bibinfo {author} {\bibfnamefont {IG}~\bibnamefont {Tamm}},\ }\bibfield  {title} {\enquote {\bibinfo {title} {The uncertainty relation between energy and time in non-relativistic quantum mechanics},}\ }\href {https://doi.org/10.1007/978-3-642-74626-0_8} {\bibfield  {journal} {\bibinfo  {journal} {J. Phys. (USSR)}\ }\textbf {\bibinfo {volume} {9}},\ \bibinfo {pages} {249} (\bibinfo {year} {1945})}\BibitemShut {NoStop}%
\bibitem [{\citenamefont {Margolus}\ and\ \citenamefont {Levitin}(1998)}]{Margolus1998}%
  \BibitemOpen
  \bibfield  {author} {\bibinfo {author} {\bibfnamefont {Norman}\ \bibnamefont {Margolus}}\ and\ \bibinfo {author} {\bibfnamefont {Lev~B.}\ \bibnamefont {Levitin}},\ }\bibfield  {title} {\enquote {\bibinfo {title} {The maximum speed of dynamical evolution},}\ }\href {\doibase https://doi.org/10.1016/S0167-2789(98)00054-2} {\bibfield  {journal} {\bibinfo  {journal} {Physica D: Nonlinear Phenomena}\ }\textbf {\bibinfo {volume} {120}},\ \bibinfo {pages} {188--195} (\bibinfo {year} {1998})}\BibitemShut {NoStop}%
\bibitem [{\citenamefont {Aharonov}\ and\ \citenamefont {Bohm}(1961)}]{Aharonov1961}%
  \BibitemOpen
  \bibfield  {author} {\bibinfo {author} {\bibfnamefont {Y.}~\bibnamefont {Aharonov}}\ and\ \bibinfo {author} {\bibfnamefont {D.}~\bibnamefont {Bohm}},\ }\bibfield  {title} {\enquote {\bibinfo {title} {Time in the quantum theory and the uncertainty relation for time and energy},}\ }\href {\doibase 10.1103/PhysRev.122.1649} {\bibfield  {journal} {\bibinfo  {journal} {Physical Review}\ }\textbf {\bibinfo {volume} {122}},\ \bibinfo {pages} {1649--1658} (\bibinfo {year} {1961})}\BibitemShut {NoStop}%
\bibitem [{\citenamefont {Anandan}\ and\ \citenamefont {Aharonov}(1990)}]{Anandan1990}%
  \BibitemOpen
  \bibfield  {author} {\bibinfo {author} {\bibfnamefont {J.}~\bibnamefont {Anandan}}\ and\ \bibinfo {author} {\bibfnamefont {Y.}~\bibnamefont {Aharonov}},\ }\bibfield  {title} {\enquote {\bibinfo {title} {Geometry of quantum evolution},}\ }\href {\doibase 10.1103/PhysRevLett.65.1697} {\bibfield  {journal} {\bibinfo  {journal} {Physical Review Letters}\ }\textbf {\bibinfo {volume} {65}},\ \bibinfo {pages} {1697--1700} (\bibinfo {year} {1990})}\BibitemShut {NoStop}%
\bibitem [{\citenamefont {del Campo}\ \emph {et~al.}(2013)\citenamefont {del Campo}, \citenamefont {Egusquiza}, \citenamefont {Plenio},\ and\ \citenamefont {Huelga}}]{Campo2013}%
  \BibitemOpen
  \bibfield  {author} {\bibinfo {author} {\bibfnamefont {A.}~\bibnamefont {del Campo}}, \bibinfo {author} {\bibfnamefont {I.~L.}\ \bibnamefont {Egusquiza}}, \bibinfo {author} {\bibfnamefont {M.~B.}\ \bibnamefont {Plenio}}, \ and\ \bibinfo {author} {\bibfnamefont {S.~F.}\ \bibnamefont {Huelga}},\ }\bibfield  {title} {\enquote {\bibinfo {title} {Quantum speed limits in open system dynamics},}\ }\href {\doibase 10.1103/PhysRevLett.110.050403} {\bibfield  {journal} {\bibinfo  {journal} {Physical Review Letters}\ }\textbf {\bibinfo {volume} {110}},\ \bibinfo {pages} {050403} (\bibinfo {year} {2013})}\BibitemShut {NoStop}%
\bibitem [{\citenamefont {Impens}\ \emph {et~al.}(2021)\citenamefont {Impens}, \citenamefont {D'Angelis}, \citenamefont {Pinheiro},\ and\ \citenamefont {Gu\'ery-Odelin}}]{Impens2021}%
  \BibitemOpen
  \bibfield  {author} {\bibinfo {author} {\bibfnamefont {F.}~\bibnamefont {Impens}}, \bibinfo {author} {\bibfnamefont {F.~M.}\ \bibnamefont {D'Angelis}}, \bibinfo {author} {\bibfnamefont {F.~A.}\ \bibnamefont {Pinheiro}}, \ and\ \bibinfo {author} {\bibfnamefont {D.}~\bibnamefont {Gu\'ery-Odelin}},\ }\bibfield  {title} {\enquote {\bibinfo {title} {Time scaling and quantum speed limit in non-hermitian hamiltonians},}\ }\href {\doibase 10.1103/PhysRevA.104.052620} {\bibfield  {journal} {\bibinfo  {journal} {Physical Review A}\ }\textbf {\bibinfo {volume} {104}},\ \bibinfo {pages} {052620} (\bibinfo {year} {2021})}\BibitemShut {NoStop}%
\bibitem [{\citenamefont {Brody}\ and\ \citenamefont {Longstaff}(2019)}]{Brody2019}%
  \BibitemOpen
  \bibfield  {author} {\bibinfo {author} {\bibfnamefont {Dorje~C.}\ \bibnamefont {Brody}}\ and\ \bibinfo {author} {\bibfnamefont {Bradley}\ \bibnamefont {Longstaff}},\ }\bibfield  {title} {\enquote {\bibinfo {title} {Evolution speed of open quantum dynamics},}\ }\href {\doibase 10.1103/PhysRevResearch.1.033127} {\bibfield  {journal} {\bibinfo  {journal} {Physical Review Research}\ }\textbf {\bibinfo {volume} {1}},\ \bibinfo {pages} {033127} (\bibinfo {year} {2019})}\BibitemShut {NoStop}%
\bibitem [{\citenamefont {Alipour}\ \emph {et~al.}(2020)\citenamefont {Alipour}, \citenamefont {Chenu}, \citenamefont {Rezakhani},\ and\ \citenamefont {del Campo}}]{Alipour2020}%
  \BibitemOpen
  \bibfield  {author} {\bibinfo {author} {\bibfnamefont {Sahar}\ \bibnamefont {Alipour}}, \bibinfo {author} {\bibfnamefont {Aurelia}\ \bibnamefont {Chenu}}, \bibinfo {author} {\bibfnamefont {Ali~T.}\ \bibnamefont {Rezakhani}}, \ and\ \bibinfo {author} {\bibfnamefont {Adolfo}\ \bibnamefont {del Campo}},\ }\bibfield  {title} {\enquote {\bibinfo {title} {Shortcuts to {A}diabaticity in {D}riven {O}pen {Q}uantum {S}ystems: {B}alanced {G}ain and {L}oss and {N}on-{M}arkovian {E}volution},}\ }\href {\doibase 10.22331/q-2020-09-28-336} {\bibfield  {journal} {\bibinfo  {journal} {{Quantum}}\ }\textbf {\bibinfo {volume} {4}},\ \bibinfo {pages} {336} (\bibinfo {year} {2020})}\BibitemShut {NoStop}%
\bibitem [{\citenamefont {Uzdin}\ \emph {et~al.}(2012)\citenamefont {Uzdin}, \citenamefont {Günther}, \citenamefont {Rahav},\ and\ \citenamefont {Moiseyev}}]{Uzdin2012}%
  \BibitemOpen
  \bibfield  {author} {\bibinfo {author} {\bibfnamefont {Raam}\ \bibnamefont {Uzdin}}, \bibinfo {author} {\bibfnamefont {Uwe}\ \bibnamefont {Günther}}, \bibinfo {author} {\bibfnamefont {Saar}\ \bibnamefont {Rahav}}, \ and\ \bibinfo {author} {\bibfnamefont {Nimrod}\ \bibnamefont {Moiseyev}},\ }\bibfield  {title} {\enquote {\bibinfo {title} {Time-dependent hamiltonians with 100\% evolution speed efficiency},}\ }\href {\doibase 10.1088/1751-8113/45/41/415304} {\bibfield  {journal} {\bibinfo  {journal} {Journal of Physics A: Mathematical and Theoretical}\ }\textbf {\bibinfo {volume} {45}},\ \bibinfo {pages} {415304} (\bibinfo {year} {2012})}\BibitemShut {NoStop}%
\bibitem [{\citenamefont {Thakuria}\ \emph {et~al.}(2023)\citenamefont {Thakuria}, \citenamefont {Srivastav}, \citenamefont {Mohan}, \citenamefont {Kumari},\ and\ \citenamefont {Pati}}]{Dimpi2022}%
  \BibitemOpen
  \bibfield  {author} {\bibinfo {author} {\bibfnamefont {Dimpi}\ \bibnamefont {Thakuria}}, \bibinfo {author} {\bibfnamefont {Abhay}\ \bibnamefont {Srivastav}}, \bibinfo {author} {\bibfnamefont {Brij}\ \bibnamefont {Mohan}}, \bibinfo {author} {\bibfnamefont {Asmita}\ \bibnamefont {Kumari}}, \ and\ \bibinfo {author} {\bibfnamefont {Arun~Kumar}\ \bibnamefont {Pati}},\ }\bibfield  {title} {\enquote {\bibinfo {title} {Generalised quantum speed limit for arbitrary time-continuous evolution},}\ }\href {\doibase 10.1088/1751-8121/ad15ad} {\bibfield  {journal} {\bibinfo  {journal} {Journal of Physics A: Mathematical and Theoretical}\ }\textbf {\bibinfo {volume} {57}},\ \bibinfo {pages} {025302} (\bibinfo {year} {2023})}\BibitemShut {NoStop}%
\bibitem [{\citenamefont {(Anthony)~Chen}\ \emph {et~al.}(2023)\citenamefont {(Anthony)~Chen}, \citenamefont {Lucas},\ and\ \citenamefont {Yin}}]{Anthony-23}%
  \BibitemOpen
  \bibfield  {author} {\bibinfo {author} {\bibfnamefont {Chi-Fang}\ \bibnamefont {(Anthony)~Chen}}, \bibinfo {author} {\bibfnamefont {Andrew}\ \bibnamefont {Lucas}}, \ and\ \bibinfo {author} {\bibfnamefont {Chao}\ \bibnamefont {Yin}},\ }\bibfield  {title} {\enquote {\bibinfo {title} {Speed limits and locality in many-body quantum dynamics},}\ }\href {\doibase 10.1088/1361-6633/acfaae} {\bibfield  {journal} {\bibinfo  {journal} {Reports on Progress in Physics}\ }\textbf {\bibinfo {volume} {86}},\ \bibinfo {pages} {116001} (\bibinfo {year} {2023})}\BibitemShut {NoStop}%
\bibitem [{\citenamefont {Zhang}\ \emph {et~al.}(2023)\citenamefont {Zhang}, \citenamefont {Yu},\ and\ \citenamefont {Liu}}]{Zhang-23}%
  \BibitemOpen
  \bibfield  {author} {\bibinfo {author} {\bibfnamefont {Mao}\ \bibnamefont {Zhang}}, \bibinfo {author} {\bibfnamefont {Huai-Ming}\ \bibnamefont {Yu}}, \ and\ \bibinfo {author} {\bibfnamefont {Jing}\ \bibnamefont {Liu}},\ }\bibfield  {title} {\enquote {\bibinfo {title} {Speed limits and locality in many-body quantum dynamics},}\ }\href {\doibase 10.1038/s41534-023-00768-8} {\bibfield  {journal} {\bibinfo  {journal} {npj Quantum Information}\ }\textbf {\bibinfo {volume} {9}},\ \bibinfo {pages} {97} (\bibinfo {year} {2023})}\BibitemShut {NoStop}%
\bibitem [{\citenamefont {Shrimali}\ \emph {et~al.}(2024)\citenamefont {Shrimali}, \citenamefont {Panda},\ and\ \citenamefont {Pati}}]{Shrimali2024}%
  \BibitemOpen
  \bibfield  {author} {\bibinfo {author} {\bibfnamefont {Divyansh}\ \bibnamefont {Shrimali}}, \bibinfo {author} {\bibfnamefont {Biswaranjan}\ \bibnamefont {Panda}}, \ and\ \bibinfo {author} {\bibfnamefont {Arun~Kumar}\ \bibnamefont {Pati}},\ }\bibfield  {title} {\enquote {\bibinfo {title} {Stronger speed limit for observables: Tighter bound for the capacity of entanglement, the modular hamiltonian, and the charging of a quantum battery},}\ }\href {\doibase 10.1103/PhysRevA.110.022425} {\bibfield  {journal} {\bibinfo  {journal} {Phys. Rev. A}\ }\textbf {\bibinfo {volume} {110}},\ \bibinfo {pages} {022425} (\bibinfo {year} {2024})}\BibitemShut {NoStop}%
\bibitem [{\citenamefont {Caneva}\ \emph {et~al.}(2009)\citenamefont {Caneva}, \citenamefont {Murphy}, \citenamefont {Calarco}, \citenamefont {Fazio}, \citenamefont {Montangero}, \citenamefont {Giovannetti},\ and\ \citenamefont {Santoro}}]{Caneva-09}%
  \BibitemOpen
  \bibfield  {author} {\bibinfo {author} {\bibfnamefont {T.}~\bibnamefont {Caneva}}, \bibinfo {author} {\bibfnamefont {M.}~\bibnamefont {Murphy}}, \bibinfo {author} {\bibfnamefont {T.}~\bibnamefont {Calarco}}, \bibinfo {author} {\bibfnamefont {R.}~\bibnamefont {Fazio}}, \bibinfo {author} {\bibfnamefont {S.}~\bibnamefont {Montangero}}, \bibinfo {author} {\bibfnamefont {V.}~\bibnamefont {Giovannetti}}, \ and\ \bibinfo {author} {\bibfnamefont {G.~E.}\ \bibnamefont {Santoro}},\ }\bibfield  {title} {\enquote {\bibinfo {title} {Optimal control at the quantum speed limit},}\ }\href {\doibase 10.1103/PhysRevLett.103.240501} {\bibfield  {journal} {\bibinfo  {journal} {Physical Review Letters}\ }\textbf {\bibinfo {volume} {103}},\ \bibinfo {pages} {240501} (\bibinfo {year} {2009})}\BibitemShut {NoStop}%
\bibitem [{\citenamefont {Campbell}\ and\ \citenamefont {Deffner}(2017)}]{Campbell-17}%
  \BibitemOpen
  \bibfield  {author} {\bibinfo {author} {\bibfnamefont {Steve}\ \bibnamefont {Campbell}}\ and\ \bibinfo {author} {\bibfnamefont {Sebastian}\ \bibnamefont {Deffner}},\ }\bibfield  {title} {\enquote {\bibinfo {title} {Trade-off between speed and cost in shortcuts to adiabaticity},}\ }\href {\doibase 10.1103/PhysRevLett.118.100601} {\bibfield  {journal} {\bibinfo  {journal} {Physical Review Letters}\ }\textbf {\bibinfo {volume} {118}},\ \bibinfo {pages} {100601} (\bibinfo {year} {2017})}\BibitemShut {NoStop}%
\bibitem [{\citenamefont {Evangelakos}\ \emph {et~al.}(2023)\citenamefont {Evangelakos}, \citenamefont {Paspalakis},\ and\ \citenamefont {Stefanatos}}]{Evangelakos-23}%
  \BibitemOpen
  \bibfield  {author} {\bibinfo {author} {\bibfnamefont {Vasileios}\ \bibnamefont {Evangelakos}}, \bibinfo {author} {\bibfnamefont {Emmanuel}\ \bibnamefont {Paspalakis}}, \ and\ \bibinfo {author} {\bibfnamefont {Dionisis}\ \bibnamefont {Stefanatos}},\ }\bibfield  {title} {\enquote {\bibinfo {title} {Minimum-time generation of a uniform superposition in a qubit with only transverse field control},}\ }\href {\doibase 10.1103/PhysRevA.108.062425} {\bibfield  {journal} {\bibinfo  {journal} {Phys. Rev. A}\ }\textbf {\bibinfo {volume} {108}},\ \bibinfo {pages} {062425} (\bibinfo {year} {2023})}\BibitemShut {NoStop}%
\bibitem [{\citenamefont {Ashhab}\ \emph {et~al.}(2012)\citenamefont {Ashhab}, \citenamefont {de~Groot},\ and\ \citenamefont {Nori}}]{Ashhab-12}%
  \BibitemOpen
  \bibfield  {author} {\bibinfo {author} {\bibfnamefont {S.}~\bibnamefont {Ashhab}}, \bibinfo {author} {\bibfnamefont {P.~C.}\ \bibnamefont {de~Groot}}, \ and\ \bibinfo {author} {\bibfnamefont {Franco}\ \bibnamefont {Nori}},\ }\bibfield  {title} {\enquote {\bibinfo {title} {Speed limits for quantum gates in multiqubit systems},}\ }\href {\doibase 10.1103/PhysRevA.85.052327} {\bibfield  {journal} {\bibinfo  {journal} {Physical Review A}\ }\textbf {\bibinfo {volume} {85}},\ \bibinfo {pages} {052327} (\bibinfo {year} {2012})}\BibitemShut {NoStop}%
\bibitem [{\citenamefont {Mohan}\ \emph {et~al.}(2022{\natexlab{a}})\citenamefont {Mohan}, \citenamefont {Das},\ and\ \citenamefont {Pati}}]{Mohan-22}%
  \BibitemOpen
  \bibfield  {author} {\bibinfo {author} {\bibfnamefont {Brij}\ \bibnamefont {Mohan}}, \bibinfo {author} {\bibfnamefont {Siddhartha}\ \bibnamefont {Das}}, \ and\ \bibinfo {author} {\bibfnamefont {Arun~Kumar}\ \bibnamefont {Pati}},\ }\bibfield  {title} {\enquote {\bibinfo {title} {Quantum speed limits for information and coherence},}\ }\href {\doibase 10.1088/1367-2630/ac753c} {\bibfield  {journal} {\bibinfo  {journal} {New Journal of Physics}\ }\textbf {\bibinfo {volume} {24}},\ \bibinfo {pages} {065003} (\bibinfo {year} {2022}{\natexlab{a}})}\BibitemShut {NoStop}%
\bibitem [{\citenamefont {Aifer}\ and\ \citenamefont {Deffner}(2022)}]{Aifer-22}%
  \BibitemOpen
  \bibfield  {author} {\bibinfo {author} {\bibfnamefont {Maxwell}\ \bibnamefont {Aifer}}\ and\ \bibinfo {author} {\bibfnamefont {Sebastian}\ \bibnamefont {Deffner}},\ }\bibfield  {title} {\enquote {\bibinfo {title} {From quantum speed limits to energy-efficient quantum gates},}\ }\href {\doibase 10.1088/1367-2630/ac6821} {\bibfield  {journal} {\bibinfo  {journal} {New Journal of Physics}\ }\textbf {\bibinfo {volume} {24}},\ \bibinfo {pages} {055002} (\bibinfo {year} {2022})}\BibitemShut {NoStop}%
\bibitem [{\citenamefont {Bennett}\ and\ \citenamefont {Wiesner}(1992)}]{Wiesner-92}%
  \BibitemOpen
  \bibfield  {author} {\bibinfo {author} {\bibfnamefont {Charles~H.}\ \bibnamefont {Bennett}}\ and\ \bibinfo {author} {\bibfnamefont {Stephen~J.}\ \bibnamefont {Wiesner}},\ }\bibfield  {title} {\enquote {\bibinfo {title} {Communication via one- and two-particle operators on {E}instein-{P}odolsky-{R}osen states},}\ }\href {\doibase 10.1103/PhysRevLett.69.2881} {\bibfield  {journal} {\bibinfo  {journal} {Physical Review Letters}\ }\textbf {\bibinfo {volume} {69}},\ \bibinfo {pages} {2881--2884} (\bibinfo {year} {1992})}\BibitemShut {NoStop}%
\bibitem [{\citenamefont {Mohan}\ and\ \citenamefont {Pati}(2021)}]{Mohan2021}%
  \BibitemOpen
  \bibfield  {author} {\bibinfo {author} {\bibfnamefont {Brij}\ \bibnamefont {Mohan}}\ and\ \bibinfo {author} {\bibfnamefont {Arun~K.}\ \bibnamefont {Pati}},\ }\bibfield  {title} {\enquote {\bibinfo {title} {Reverse quantum speed limit: How slowly a quantum battery can discharge},}\ }\href {\doibase 10.1103/PhysRevA.104.042209} {\bibfield  {journal} {\bibinfo  {journal} {Physical Review A}\ }\textbf {\bibinfo {volume} {104}},\ \bibinfo {pages} {042209} (\bibinfo {year} {2021})}\BibitemShut {NoStop}%
\bibitem [{\citenamefont {Das}\ \emph {et~al.}(2023)\citenamefont {Das}, \citenamefont {Mahunta}, \citenamefont {Agarwalla},\ and\ \citenamefont {Mukherjee}}]{Victor-23}%
  \BibitemOpen
  \bibfield  {author} {\bibinfo {author} {\bibfnamefont {Arpan}\ \bibnamefont {Das}}, \bibinfo {author} {\bibfnamefont {Shishira}\ \bibnamefont {Mahunta}}, \bibinfo {author} {\bibfnamefont {Bijay~Kumar}\ \bibnamefont {Agarwalla}}, \ and\ \bibinfo {author} {\bibfnamefont {Victor}\ \bibnamefont {Mukherjee}},\ }\bibfield  {title} {\enquote {\bibinfo {title} {Precision bound and optimal control in periodically modulated continuous quantum thermal machines},}\ }\href {\doibase 10.1103/PhysRevE.108.014137} {\bibfield  {journal} {\bibinfo  {journal} {Phys. Rev. E}\ }\textbf {\bibinfo {volume} {108}},\ \bibinfo {pages} {014137} (\bibinfo {year} {2023})}\BibitemShut {NoStop}%
\bibitem [{\citenamefont {Bhandari}\ and\ \citenamefont {Jordan}(2022)}]{Bhandari-22}%
  \BibitemOpen
  \bibfield  {author} {\bibinfo {author} {\bibfnamefont {Bibek}\ \bibnamefont {Bhandari}}\ and\ \bibinfo {author} {\bibfnamefont {Andrew~N.}\ \bibnamefont {Jordan}},\ }\bibfield  {title} {\enquote {\bibinfo {title} {Continuous measurement boosted adiabatic quantum thermal machines},}\ }\href {\doibase 10.1103/PhysRevResearch.4.033103} {\bibfield  {journal} {\bibinfo  {journal} {Phys. Rev. Res.}\ }\textbf {\bibinfo {volume} {4}},\ \bibinfo {pages} {033103} (\bibinfo {year} {2022})}\BibitemShut {NoStop}%
\bibitem [{\citenamefont {Xu}\ and\ \citenamefont {Swingle}(2024)}]{Xu2024}%
  \BibitemOpen
  \bibfield  {author} {\bibinfo {author} {\bibfnamefont {Shenglong}\ \bibnamefont {Xu}}\ and\ \bibinfo {author} {\bibfnamefont {Brian}\ \bibnamefont {Swingle}},\ }\bibfield  {title} {\enquote {\bibinfo {title} {Scrambling dynamics and out-of-time-ordered correlators in quantum many-body systems},}\ }\href {\doibase 10.1103/PRXQuantum.5.010201} {\bibfield  {journal} {\bibinfo  {journal} {PRX Quantum}\ }\textbf {\bibinfo {volume} {5}},\ \bibinfo {pages} {010201} (\bibinfo {year} {2024})}\BibitemShut {NoStop}%
\bibitem [{\citenamefont {Jensen}(1992)}]{Jensen1992}%
  \BibitemOpen
  \bibfield  {author} {\bibinfo {author} {\bibfnamefont {Roderick~V.}\ \bibnamefont {Jensen}},\ }\bibfield  {title} {\enquote {\bibinfo {title} {Quantum chaos},}\ }\href {\doibase 10.1038/355311a0} {\bibfield  {journal} {\bibinfo  {journal} {Nature}\ }\textbf {\bibinfo {volume} {355}},\ \bibinfo {pages} {311--318} (\bibinfo {year} {1992})}\BibitemShut {NoStop}%
\bibitem [{\citenamefont {Heyl}(2018)}]{heyl2018}%
  \BibitemOpen
  \bibfield  {author} {\bibinfo {author} {\bibfnamefont {Markus}\ \bibnamefont {Heyl}},\ }\bibfield  {title} {\enquote {\bibinfo {title} {Dynamical quantum phase transitions: a review},}\ }\href@noop {} {\bibfield  {journal} {\bibinfo  {journal} {Reports on Progress in Physics}\ }\textbf {\bibinfo {volume} {81}},\ \bibinfo {pages} {054001} (\bibinfo {year} {2018})}\BibitemShut {NoStop}%
\bibitem [{\citenamefont {Vojta}(2003)}]{Vojta2003}%
  \BibitemOpen
  \bibfield  {author} {\bibinfo {author} {\bibfnamefont {Matthias}\ \bibnamefont {Vojta}},\ }\bibfield  {title} {\enquote {\bibinfo {title} {Quantum phase transitions},}\ }\href@noop {} {\bibfield  {journal} {\bibinfo  {journal} {Reports on Progress in Physics}\ }\textbf {\bibinfo {volume} {66}},\ \bibinfo {pages} {2069} (\bibinfo {year} {2003})}\BibitemShut {NoStop}%
\bibitem [{\citenamefont {van Vliet}(1979)}]{Vliet1979}%
  \BibitemOpen
  \bibfield  {author} {\bibinfo {author} {\bibfnamefont {K.~M.}\ \bibnamefont {van Vliet}},\ }\bibfield  {title} {\enquote {\bibinfo {title} {{Linear response theory revisited. II. The master equation approach}},}\ }\href {\doibase 10.1063/1.524020} {\bibfield  {journal} {\bibinfo  {journal} {Journal of Mathematical Physics}\ }\textbf {\bibinfo {volume} {20}},\ \bibinfo {pages} {2573--2595} (\bibinfo {year} {1979})},\ \Eprint {http://arxiv.org/abs/https://pubs.aip.org/aip/jmp/article-pdf/20/12/2573/19306554/2573\_1\_online.pdf} {https://pubs.aip.org/aip/jmp/article-pdf/20/12/2573/19306554/2573\_1\_online.pdf} \BibitemShut {NoStop}%
\bibitem [{\citenamefont {Jing}\ \emph {et~al.}(2016)\citenamefont {Jing}, \citenamefont {Wu},\ and\ \citenamefont {del Campo}}]{Jing2016}%
  \BibitemOpen
  \bibfield  {author} {\bibinfo {author} {\bibfnamefont {Jun}\ \bibnamefont {Jing}}, \bibinfo {author} {\bibfnamefont {Lian-Ao}\ \bibnamefont {Wu}}, \ and\ \bibinfo {author} {\bibfnamefont {Adolfo}\ \bibnamefont {del Campo}},\ }\bibfield  {title} {\enquote {\bibinfo {title} {Fundamental speed limits to the generation of quantumness},}\ }\href {\doibase 10.1038/srep38149} {\bibfield  {journal} {\bibinfo  {journal} {Scientific Reports}\ }\textbf {\bibinfo {volume} {6}},\ \bibinfo {pages} {38149} (\bibinfo {year} {2016})}\BibitemShut {NoStop}%
\bibitem [{\citenamefont {Luo}(2003)}]{Luo-03}%
  \BibitemOpen
  \bibfield  {author} {\bibinfo {author} {\bibfnamefont {Shunlong}\ \bibnamefont {Luo}},\ }\bibfield  {title} {\enquote {\bibinfo {title} {Wigner-yanase skew information and uncertainty relations},}\ }\href {\doibase 10.1103/PhysRevLett.91.180403} {\bibfield  {journal} {\bibinfo  {journal} {Phys. Rev. Lett.}\ }\textbf {\bibinfo {volume} {91}},\ \bibinfo {pages} {180403} (\bibinfo {year} {2003})}\BibitemShut {NoStop}%
\bibitem [{\citenamefont {Hu}\ and\ \citenamefont {Fan}(2017)}]{Hu-2017}%
  \BibitemOpen
  \bibfield  {author} {\bibinfo {author} {\bibfnamefont {Ming-Liang}\ \bibnamefont {Hu}}\ and\ \bibinfo {author} {\bibfnamefont {Heng}\ \bibnamefont {Fan}},\ }\bibfield  {title} {\enquote {\bibinfo {title} {Relative quantum coherence, incompatibility, and quantum correlations of states},}\ }\href {\doibase 10.1103/physreva.95.052106} {\bibfield  {journal} {\bibinfo  {journal} {Physical Review A}\ }\textbf {\bibinfo {volume} {95}} (\bibinfo {year} {2017}),\ 10.1103/physreva.95.052106}\BibitemShut {NoStop}%
\bibitem [{\citenamefont {Bu}\ \emph {et~al.}(2017)\citenamefont {Bu}, \citenamefont {Singh}, \citenamefont {Fei}, \citenamefont {Pati},\ and\ \citenamefont {Wu}}]{Bu-2017}%
  \BibitemOpen
  \bibfield  {author} {\bibinfo {author} {\bibfnamefont {Kaifeng}\ \bibnamefont {Bu}}, \bibinfo {author} {\bibfnamefont {Uttam}\ \bibnamefont {Singh}}, \bibinfo {author} {\bibfnamefont {Shao-Ming}\ \bibnamefont {Fei}}, \bibinfo {author} {\bibfnamefont {Arun~Kumar}\ \bibnamefont {Pati}}, \ and\ \bibinfo {author} {\bibfnamefont {Junde}\ \bibnamefont {Wu}},\ }\bibfield  {title} {\enquote {\bibinfo {title} {Maximum relative entropy of coherence: An operational coherence measure},}\ }\href {\doibase 10.1103/physrevlett.119.150405} {\bibfield  {journal} {\bibinfo  {journal} {Physical Review Letters}\ }\textbf {\bibinfo {volume} {119}} (\bibinfo {year} {2017}),\ 10.1103/physrevlett.119.150405}\BibitemShut {NoStop}%
\bibitem [{\citenamefont {Rivas}\ and\ \citenamefont {Huelga}(2012)}]{Rivas-12}%
  \BibitemOpen
  \bibfield  {author} {\bibinfo {author} {\bibfnamefont {Angel}\ \bibnamefont {Rivas}}\ and\ \bibinfo {author} {\bibfnamefont {Susana~F}\ \bibnamefont {Huelga}},\ }\href@noop {} {\emph {\bibinfo {title} {Open quantum systems}}},\ Vol.~\bibinfo {volume} {10}\ (\bibinfo  {publisher} {Springer},\ \bibinfo {year} {2012})\BibitemShut {NoStop}%
\bibitem [{\citenamefont {Ma}\ \emph {et~al.}(2014)\citenamefont {Ma}, \citenamefont {Zhao}, \citenamefont {Wang},\ and\ \citenamefont {Fei}}]{Ma2014}%
  \BibitemOpen
  \bibfield  {author} {\bibinfo {author} {\bibfnamefont {Teng}\ \bibnamefont {Ma}}, \bibinfo {author} {\bibfnamefont {Ming-Jing}\ \bibnamefont {Zhao}}, \bibinfo {author} {\bibfnamefont {Yao-Kun}\ \bibnamefont {Wang}}, \ and\ \bibinfo {author} {\bibfnamefont {Shao-Ming}\ \bibnamefont {Fei}},\ }\bibfield  {title} {\enquote {\bibinfo {title} {Non-commutativity and local indistinguishability of quantum states},}\ }\href {\doibase 10.1038/srep06336} {\bibfield  {journal} {\bibinfo  {journal} {Scientific Reports}\ }\textbf {\bibinfo {volume} {4}},\ \bibinfo {pages} {6336} (\bibinfo {year} {2014})}\BibitemShut {NoStop}%
\bibitem [{\citenamefont {Guo}(2016)}]{Guo2016}%
  \BibitemOpen
  \bibfield  {author} {\bibinfo {author} {\bibfnamefont {Yu}~\bibnamefont {Guo}},\ }\bibfield  {title} {\enquote {\bibinfo {title} {Non-commutativity measure of quantum discord},}\ }\href {https://www.nature.com/articles/srep25241#citeas} {\bibfield  {journal} {\bibinfo  {journal} {Scientific reports}\ }\textbf {\bibinfo {volume} {6}},\ \bibinfo {pages} {1--8} (\bibinfo {year} {2016})}\BibitemShut {NoStop}%
\bibitem [{\citenamefont {Mordasewicz}\ and\ \citenamefont {Kaniewski}(2022)}]{Mord2022}%
  \BibitemOpen
  \bibfield  {author} {\bibinfo {author} {\bibfnamefont {Krzysztof}\ \bibnamefont {Mordasewicz}}\ and\ \bibinfo {author} {\bibfnamefont {Jędrzej}\ \bibnamefont {Kaniewski}},\ }\bibfield  {title} {\enquote {\bibinfo {title} {Quantifying incompatibility of quantum measurements through non-commutativity},}\ }\href {\doibase 10.1088/1751-8121/ac71eb} {\bibfield  {journal} {\bibinfo  {journal} {Journal of Physics A: Mathematical and Theoretical}\ }\textbf {\bibinfo {volume} {55}},\ \bibinfo {pages} {265302} (\bibinfo {year} {2022})}\BibitemShut {NoStop}%
\bibitem [{\citenamefont {Iyengar}\ \emph {et~al.}(2013)\citenamefont {Iyengar}, \citenamefont {Chandan},\ and\ \citenamefont {Srikanth}}]{Iyengar2013}%
  \BibitemOpen
  \bibfield  {author} {\bibinfo {author} {\bibfnamefont {Pavan}\ \bibnamefont {Iyengar}}, \bibinfo {author} {\bibfnamefont {G.~N.}\ \bibnamefont {Chandan}}, \ and\ \bibinfo {author} {\bibfnamefont {R.}~\bibnamefont {Srikanth}},\ }\bibfield  {title} {\enquote {\bibinfo {title} {Quantifying quantumness via commutators: an application to quantum walk},}\ }\href@noop {} {\  (\bibinfo {year} {2013})},\ \Eprint {http://arxiv.org/abs/1312.1329} {arXiv:1312.1329 [quant-ph]} \BibitemShut {NoStop}%
\bibitem [{\citenamefont {Ferro}\ \emph {et~al.}(2018)\citenamefont {Ferro}, \citenamefont {Fazio}, \citenamefont {Illuminati}, \citenamefont {Marmo}, \citenamefont {Pascazio},\ and\ \citenamefont {Vedral}}]{Ferro2018}%
  \BibitemOpen
  \bibfield  {author} {\bibinfo {author} {\bibfnamefont {Leonardo}\ \bibnamefont {Ferro}}, \bibinfo {author} {\bibfnamefont {Rosario}\ \bibnamefont {Fazio}}, \bibinfo {author} {\bibfnamefont {Fabrizio}\ \bibnamefont {Illuminati}}, \bibinfo {author} {\bibfnamefont {Giuseppe}\ \bibnamefont {Marmo}}, \bibinfo {author} {\bibfnamefont {Saverio}\ \bibnamefont {Pascazio}}, \ and\ \bibinfo {author} {\bibfnamefont {Vlatko}\ \bibnamefont {Vedral}},\ }\bibfield  {title} {\enquote {\bibinfo {title} {Measuring quantumness: from theory to observability in interferometric setups},}\ }\href {\doibase 10.1140/epjd/e2018-90522-y} {\bibfield  {journal} {\bibinfo  {journal} {The European Physical Journal D}\ }\textbf {\bibinfo {volume} {72}},\ \bibinfo {pages} {219} (\bibinfo {year} {2018})}\BibitemShut {NoStop}%
\bibitem [{\citenamefont {Luo}(2006)}]{Luo2006}%
  \BibitemOpen
  \bibfield  {author} {\bibinfo {author} {\bibfnamefont {Shunlong}\ \bibnamefont {Luo}},\ }\bibfield  {title} {\enquote {\bibinfo {title} {Quantum uncertainty of mixed states based on skew information},}\ }\href {\doibase 10.1103/PhysRevA.73.022324} {\bibfield  {journal} {\bibinfo  {journal} {Phys. Rev. A}\ }\textbf {\bibinfo {volume} {73}},\ \bibinfo {pages} {022324} (\bibinfo {year} {2006})}\BibitemShut {NoStop}%
\bibitem [{\citenamefont {Luo}(2004)}]{Luo2004}%
  \BibitemOpen
  \bibfield  {author} {\bibinfo {author} {\bibfnamefont {Shunlong}\ \bibnamefont {Luo}},\ }\bibfield  {title} {\enquote {\bibinfo {title} {Wigner-yanase skew information vs. quantum fisher information},}\ }\href@noop {} {\bibfield  {journal} {\bibinfo  {journal} {Proceedings of the American Mathematical Society}\ }\textbf {\bibinfo {volume} {132}},\ \bibinfo {pages} {885--890} (\bibinfo {year} {2004})}\BibitemShut {NoStop}%
\bibitem [{\citenamefont {Luo}\ \emph {et~al.}(2012)\citenamefont {Luo}, \citenamefont {Fu},\ and\ \citenamefont {Oh}}]{Luo2012}%
  \BibitemOpen
  \bibfield  {author} {\bibinfo {author} {\bibfnamefont {Shunlong}\ \bibnamefont {Luo}}, \bibinfo {author} {\bibfnamefont {Shuangshuang}\ \bibnamefont {Fu}}, \ and\ \bibinfo {author} {\bibfnamefont {Choo~Hiap}\ \bibnamefont {Oh}},\ }\bibfield  {title} {\enquote {\bibinfo {title} {Quantifying correlations via the wigner-yanase skew information},}\ }\href {\doibase 10.1103/PhysRevA.85.032117} {\bibfield  {journal} {\bibinfo  {journal} {Phys. Rev. A}\ }\textbf {\bibinfo {volume} {85}},\ \bibinfo {pages} {032117} (\bibinfo {year} {2012})}\BibitemShut {NoStop}%
\bibitem [{\citenamefont {Sun}\ \emph {et~al.}(2022)\citenamefont {Sun}, \citenamefont {Li},\ and\ \citenamefont {Luo}}]{Luo2022}%
  \BibitemOpen
  \bibfield  {author} {\bibinfo {author} {\bibfnamefont {Yuan}\ \bibnamefont {Sun}}, \bibinfo {author} {\bibfnamefont {Nan}\ \bibnamefont {Li}}, \ and\ \bibinfo {author} {\bibfnamefont {Shunlong}\ \bibnamefont {Luo}},\ }\bibfield  {title} {\enquote {\bibinfo {title} {Quantifying coherence relative to channels via metric-adjusted skew information},}\ }\href {\doibase 10.1103/PhysRevA.106.012436} {\bibfield  {journal} {\bibinfo  {journal} {Phys. Rev. A}\ }\textbf {\bibinfo {volume} {106}},\ \bibinfo {pages} {012436} (\bibinfo {year} {2022})}\BibitemShut {NoStop}%
\bibitem [{\citenamefont {Yu}(2017)}]{Yu2017}%
  \BibitemOpen
  \bibfield  {author} {\bibinfo {author} {\bibfnamefont {Chang-shui}\ \bibnamefont {Yu}},\ }\bibfield  {title} {\enquote {\bibinfo {title} {Quantum coherence via skew information and its polygamy},}\ }\href {\doibase 10.1103/PhysRevA.95.042337} {\bibfield  {journal} {\bibinfo  {journal} {Phys. Rev. A}\ }\textbf {\bibinfo {volume} {95}},\ \bibinfo {pages} {042337} (\bibinfo {year} {2017})}\BibitemShut {NoStop}%
\bibitem [{\citenamefont {Li}\ and\ \citenamefont {Lin}(2016)}]{Li2016}%
  \BibitemOpen
  \bibfield  {author} {\bibinfo {author} {\bibfnamefont {Yan-Chao}\ \bibnamefont {Li}}\ and\ \bibinfo {author} {\bibfnamefont {Hai-Qing}\ \bibnamefont {Lin}},\ }\bibfield  {title} {\enquote {\bibinfo {title} {Quantum coherence and quantum phase transitions},}\ }\href {\doibase 10.1038/srep26365} {\bibfield  {journal} {\bibinfo  {journal} {Scientific Reports}\ }\textbf {\bibinfo {volume} {6}},\ \bibinfo {pages} {26365} (\bibinfo {year} {2016})}\BibitemShut {NoStop}%
\bibitem [{\citenamefont {Marvian}\ \emph {et~al.}(2016)\citenamefont {Marvian}, \citenamefont {Spekkens},\ and\ \citenamefont {Zanardi}}]{Marvian-2016}%
  \BibitemOpen
  \bibfield  {author} {\bibinfo {author} {\bibfnamefont {Iman}\ \bibnamefont {Marvian}}, \bibinfo {author} {\bibfnamefont {Robert~W.}\ \bibnamefont {Spekkens}}, \ and\ \bibinfo {author} {\bibfnamefont {Paolo}\ \bibnamefont {Zanardi}},\ }\bibfield  {title} {\enquote {\bibinfo {title} {Quantum speed limits, coherence, and asymmetry},}\ }\href {\doibase 10.1103/physreva.93.052331} {\bibfield  {journal} {\bibinfo  {journal} {Physical Review A}\ }\textbf {\bibinfo {volume} {93}} (\bibinfo {year} {2016}),\ 10.1103/physreva.93.052331}\BibitemShut {NoStop}%
\bibitem [{\citenamefont {Pires}\ \emph {et~al.}(2016)\citenamefont {Pires}, \citenamefont {Cianciaruso}, \citenamefont {C\'eleri}, \citenamefont {Adesso},\ and\ \citenamefont {Soares-Pinto}}]{Pires-2016}%
  \BibitemOpen
  \bibfield  {author} {\bibinfo {author} {\bibfnamefont {Diego~Paiva}\ \bibnamefont {Pires}}, \bibinfo {author} {\bibfnamefont {Marco}\ \bibnamefont {Cianciaruso}}, \bibinfo {author} {\bibfnamefont {Lucas~C.}\ \bibnamefont {C\'eleri}}, \bibinfo {author} {\bibfnamefont {Gerardo}\ \bibnamefont {Adesso}}, \ and\ \bibinfo {author} {\bibfnamefont {Diogo~O.}\ \bibnamefont {Soares-Pinto}},\ }\bibfield  {title} {\enquote {\bibinfo {title} {Generalized geometric quantum speed limits},}\ }\href {\doibase 10.1103/PhysRevX.6.021031} {\bibfield  {journal} {\bibinfo  {journal} {Phys. Rev. X}\ }\textbf {\bibinfo {volume} {6}},\ \bibinfo {pages} {021031} (\bibinfo {year} {2016})}\BibitemShut {NoStop}%
\bibitem [{\citenamefont {Wei}(2020)}]{Wei-2020}%
  \BibitemOpen
  \bibfield  {author} {\bibinfo {author} {\bibfnamefont {Lu}~\bibnamefont {Wei}},\ }\bibfield  {title} {\enquote {\bibinfo {title} {Skewness of von neumann entanglement entropy},}\ }\href {\doibase 10.1088/1751-8121/ab63a7} {\bibfield  {journal} {\bibinfo  {journal} {Journal of Physics A: Mathematical and Theoretical}\ }\textbf {\bibinfo {volume} {53}},\ \bibinfo {pages} {075302} (\bibinfo {year} {2020})}\BibitemShut {NoStop}%
\bibitem [{\citenamefont {Streltsov}\ \emph {et~al.}(2017)\citenamefont {Streltsov}, \citenamefont {Adesso},\ and\ \citenamefont {Plenio}}]{Streltsov2017}%
  \BibitemOpen
  \bibfield  {author} {\bibinfo {author} {\bibfnamefont {Alexander}\ \bibnamefont {Streltsov}}, \bibinfo {author} {\bibfnamefont {Gerardo}\ \bibnamefont {Adesso}}, \ and\ \bibinfo {author} {\bibfnamefont {Martin~B.}\ \bibnamefont {Plenio}},\ }\bibfield  {title} {\enquote {\bibinfo {title} {Colloquium: Quantum coherence as a resource},}\ }\href {\doibase 10.1103/RevModPhys.89.041003} {\bibfield  {journal} {\bibinfo  {journal} {Rev. Mod. Phys.}\ }\textbf {\bibinfo {volume} {89}},\ \bibinfo {pages} {041003} (\bibinfo {year} {2017})}\BibitemShut {NoStop}%
\bibitem [{\citenamefont {Ahnefeld}\ \emph {et~al.}(2022)\citenamefont {Ahnefeld}, \citenamefont {Theurer}, \citenamefont {Egloff}, \citenamefont {Matera},\ and\ \citenamefont {Plenio}}]{Ahnefeld2022}%
  \BibitemOpen
  \bibfield  {author} {\bibinfo {author} {\bibfnamefont {Felix}\ \bibnamefont {Ahnefeld}}, \bibinfo {author} {\bibfnamefont {Thomas}\ \bibnamefont {Theurer}}, \bibinfo {author} {\bibfnamefont {Dario}\ \bibnamefont {Egloff}}, \bibinfo {author} {\bibfnamefont {Juan~Mauricio}\ \bibnamefont {Matera}}, \ and\ \bibinfo {author} {\bibfnamefont {Martin~B.}\ \bibnamefont {Plenio}},\ }\bibfield  {title} {\enquote {\bibinfo {title} {Coherence as a resource for shor's algorithm},}\ }\href {\doibase 10.1103/PhysRevLett.129.120501} {\bibfield  {journal} {\bibinfo  {journal} {Phys. Rev. Lett.}\ }\textbf {\bibinfo {volume} {129}},\ \bibinfo {pages} {120501} (\bibinfo {year} {2022})}\BibitemShut {NoStop}%
\bibitem [{\citenamefont {Hammam}\ \emph {et~al.}(2022)\citenamefont {Hammam}, \citenamefont {Leitch}, \citenamefont {Hassouni},\ and\ \citenamefont {Chiara}}]{Hammam2022}%
  \BibitemOpen
  \bibfield  {author} {\bibinfo {author} {\bibfnamefont {Kenza}\ \bibnamefont {Hammam}}, \bibinfo {author} {\bibfnamefont {Heather}\ \bibnamefont {Leitch}}, \bibinfo {author} {\bibfnamefont {Yassine}\ \bibnamefont {Hassouni}}, \ and\ \bibinfo {author} {\bibfnamefont {Gabriele~De}\ \bibnamefont {Chiara}},\ }\bibfield  {title} {\enquote {\bibinfo {title} {Exploiting coherence for quantum thermodynamic advantage},}\ }\href {\doibase 10.1088/1367-2630/aca49b} {\bibfield  {journal} {\bibinfo  {journal} {New Journal of Physics}\ }\textbf {\bibinfo {volume} {24}},\ \bibinfo {pages} {113053} (\bibinfo {year} {2022})}\BibitemShut {NoStop}%
\bibitem [{\citenamefont {Gour}(2022)}]{Gour2022}%
  \BibitemOpen
  \bibfield  {author} {\bibinfo {author} {\bibfnamefont {Gilad}\ \bibnamefont {Gour}},\ }\bibfield  {title} {\enquote {\bibinfo {title} {Role of quantum coherence in thermodynamics},}\ }\href {\doibase 10.1103/PRXQuantum.3.040323} {\bibfield  {journal} {\bibinfo  {journal} {PRX Quantum}\ }\textbf {\bibinfo {volume} {3}},\ \bibinfo {pages} {040323} (\bibinfo {year} {2022})}\BibitemShut {NoStop}%
\bibitem [{\citenamefont {Lostaglio}(2019)}]{Lostaglio2019}%
  \BibitemOpen
  \bibfield  {author} {\bibinfo {author} {\bibfnamefont {Matteo}\ \bibnamefont {Lostaglio}},\ }\bibfield  {title} {\enquote {\bibinfo {title} {An introductory review of the resource theory approach to thermodynamics},}\ }\href {\doibase 10.1088/1361-6633/ab46e5} {\bibfield  {journal} {\bibinfo  {journal} {Reports on Progress in Physics}\ }\textbf {\bibinfo {volume} {82}},\ \bibinfo {pages} {114001} (\bibinfo {year} {2019})}\BibitemShut {NoStop}%
\bibitem [{\citenamefont {Baumgratz}\ \emph {et~al.}(2014)\citenamefont {Baumgratz}, \citenamefont {Cramer},\ and\ \citenamefont {Plenio}}]{Baumgratz-14}%
  \BibitemOpen
  \bibfield  {author} {\bibinfo {author} {\bibfnamefont {T.}~\bibnamefont {Baumgratz}}, \bibinfo {author} {\bibfnamefont {M.}~\bibnamefont {Cramer}}, \ and\ \bibinfo {author} {\bibfnamefont {M.~B.}\ \bibnamefont {Plenio}},\ }\bibfield  {title} {\enquote {\bibinfo {title} {Quantifying coherence},}\ }\href {\doibase 10.1103/PhysRevLett.113.140401} {\bibfield  {journal} {\bibinfo  {journal} {Phys. Rev. Lett.}\ }\textbf {\bibinfo {volume} {113}},\ \bibinfo {pages} {140401} (\bibinfo {year} {2014})}\BibitemShut {NoStop}%
\bibitem [{\citenamefont {Xiong}\ and\ \citenamefont {Wu}(2018)}]{Xiong2018}%
  \BibitemOpen
  \bibfield  {author} {\bibinfo {author} {\bibfnamefont {Chunhe}\ \bibnamefont {Xiong}}\ and\ \bibinfo {author} {\bibfnamefont {Junde}\ \bibnamefont {Wu}},\ }\bibfield  {title} {\enquote {\bibinfo {title} {Geometric coherence and quantum state discrimination},}\ }\href {\doibase 10.1088/1751-8121/aac979} {\bibfield  {journal} {\bibinfo  {journal} {Journal of Physics A: Mathematical and Theoretical}\ }\textbf {\bibinfo {volume} {51}},\ \bibinfo {pages} {414005} (\bibinfo {year} {2018})}\BibitemShut {NoStop}%
\bibitem [{\citenamefont {Napoli}\ \emph {et~al.}(2016)\citenamefont {Napoli}, \citenamefont {Bromley}, \citenamefont {Cianciaruso}, \citenamefont {Piani}, \citenamefont {Johnston},\ and\ \citenamefont {Adesso}}]{Napoli2016}%
  \BibitemOpen
  \bibfield  {author} {\bibinfo {author} {\bibfnamefont {Carmine}\ \bibnamefont {Napoli}}, \bibinfo {author} {\bibfnamefont {Thomas~R.}\ \bibnamefont {Bromley}}, \bibinfo {author} {\bibfnamefont {Marco}\ \bibnamefont {Cianciaruso}}, \bibinfo {author} {\bibfnamefont {Marco}\ \bibnamefont {Piani}}, \bibinfo {author} {\bibfnamefont {Nathaniel}\ \bibnamefont {Johnston}}, \ and\ \bibinfo {author} {\bibfnamefont {Gerardo}\ \bibnamefont {Adesso}},\ }\bibfield  {title} {\enquote {\bibinfo {title} {Robustness of coherence: An operational and observable measure of quantum coherence},}\ }\href {\doibase 10.1103/PhysRevLett.116.150502} {\bibfield  {journal} {\bibinfo  {journal} {Phys. Rev. Lett.}\ }\textbf {\bibinfo {volume} {116}},\ \bibinfo {pages} {150502} (\bibinfo {year} {2016})}\BibitemShut {NoStop}%
\bibitem [{\citenamefont {Girolami}(2014)}]{Girolami-2014}%
  \BibitemOpen
  \bibfield  {author} {\bibinfo {author} {\bibfnamefont {Davide}\ \bibnamefont {Girolami}},\ }\bibfield  {title} {\enquote {\bibinfo {title} {Observable measure of quantum coherence in finite dimensional systems},}\ }\href {\doibase 10.1103/PhysRevLett.113.170401} {\bibfield  {journal} {\bibinfo  {journal} {Phys. Rev. Lett.}\ }\textbf {\bibinfo {volume} {113}},\ \bibinfo {pages} {170401} (\bibinfo {year} {2014})}\BibitemShut {NoStop}%
\bibitem [{\citenamefont {Mohan}\ \emph {et~al.}(2022{\natexlab{b}})\citenamefont {Mohan}, \citenamefont {Das},\ and\ \citenamefont {Pati}}]{Mohan2022}%
  \BibitemOpen
  \bibfield  {author} {\bibinfo {author} {\bibfnamefont {Brij}\ \bibnamefont {Mohan}}, \bibinfo {author} {\bibfnamefont {Siddhartha}\ \bibnamefont {Das}}, \ and\ \bibinfo {author} {\bibfnamefont {Arun~Kumar}\ \bibnamefont {Pati}},\ }\bibfield  {title} {\enquote {\bibinfo {title} {Quantum speed limits for information and coherence},}\ }\href {\doibase 10.1088/1367-2630/ac753c} {\bibfield  {journal} {\bibinfo  {journal} {New Journal of Physics}\ }\textbf {\bibinfo {volume} {24}},\ \bibinfo {pages} {065003} (\bibinfo {year} {2022}{\natexlab{b}})}\BibitemShut {NoStop}%
\bibitem [{\citenamefont {Mirkin}\ \emph {et~al.}(2016)\citenamefont {Mirkin}, \citenamefont {Toscano},\ and\ \citenamefont {Wisniacki}}]{Mirkin-2016}%
  \BibitemOpen
  \bibfield  {author} {\bibinfo {author} {\bibfnamefont {Nicolás}\ \bibnamefont {Mirkin}}, \bibinfo {author} {\bibfnamefont {Fabricio}\ \bibnamefont {Toscano}}, \ and\ \bibinfo {author} {\bibfnamefont {Diego~A.}\ \bibnamefont {Wisniacki}},\ }\bibfield  {title} {\enquote {\bibinfo {title} {Quantum-speed-limit bounds in an open quantum evolution},}\ }\href {\doibase 10.1103/physreva.94.052125} {\bibfield  {journal} {\bibinfo  {journal} {Physical Review A}\ }\textbf {\bibinfo {volume} {94}} (\bibinfo {year} {2016}),\ 10.1103/physreva.94.052125}\BibitemShut {NoStop}%
\bibitem [{\citenamefont {xiong Wu}\ and\ \citenamefont {shui Yu}(2020)}]{Wu-20}%
  \BibitemOpen
  \bibfield  {author} {\bibinfo {author} {\bibfnamefont {Shao}\ \bibnamefont {xiong Wu}}\ and\ \bibinfo {author} {\bibfnamefont {Chang}\ \bibnamefont {shui Yu}},\ }\bibfield  {title} {\enquote {\bibinfo {title} {Quantum speed limit based on the bound of bures angle},}\ }\href {\doibase 10.1038/s41598-020-62409-w} {\bibfield  {journal} {\bibinfo  {journal} {Scientific Reports}\ }\textbf {\bibinfo {volume} {10}} (\bibinfo {year} {2020}),\ 10.1038/s41598-020-62409-w}\BibitemShut {NoStop}%
\bibitem [{\citenamefont {Chin}\ \emph {et~al.}(2012)\citenamefont {Chin}, \citenamefont {Huelga},\ and\ \citenamefont {Plenio}}]{Chin2012}%
  \BibitemOpen
  \bibfield  {author} {\bibinfo {author} {\bibfnamefont {Alex~W.}\ \bibnamefont {Chin}}, \bibinfo {author} {\bibfnamefont {Susana~F.}\ \bibnamefont {Huelga}}, \ and\ \bibinfo {author} {\bibfnamefont {Martin~B.}\ \bibnamefont {Plenio}},\ }\bibfield  {title} {\enquote {\bibinfo {title} {Quantum metrology in non-markovian environments},}\ }\href {\doibase 10.1103/PhysRevLett.109.233601} {\bibfield  {journal} {\bibinfo  {journal} {Phys. Rev. Lett.}\ }\textbf {\bibinfo {volume} {109}},\ \bibinfo {pages} {233601} (\bibinfo {year} {2012})}\BibitemShut {NoStop}%
\end{thebibliography}%

\begin{widetext}
\appendix

\section{Proof of Corollary 2}\label{Appendix:proof_cor_2}
\begin{proof}

The skew-information at time $t$ can be expressed as

\begin{equation}
I(\rho,\cal{A}_t) = \frac{1}{2}\tr([\sqrt{\rho},\cal{A}_t]^{\dagger}[\sqrt{\rho},\cal{A}_t]).
\end{equation}

After differentiating the above equation with respect to time $t$, we obtain
\begin{equation}
     \frac{{\rm d} }{{\rm d}t}I(\rho,\cal{A}_t) = \frac{1}{2}\tr([\sqrt{\rho},\dot{\cal{A}}_t]^{\dagger}[\sqrt{\rho},\cal{A}_t])+\frac{1}{2}\tr([\sqrt{\rho},\cal{A}_t]^{\dagger}[\sqrt{\rho},\dot{\cal{A}_t}]).
\end{equation}

Let us now consider the absolute value of the above equation and apply triangular inequality $|A+B|\leq |A|+|B|$. We then obtain the following inequality

\begin{align}
     \left|\frac{{\rm d} }{{\rm d}t}I(\rho,\cal{A}_t)\right| \leq\hspace{0.2cm}& \frac{1}{2}\left|\tr([\sqrt{\rho},\dot{\cal{A}}_t]^{\dagger}[\sqrt{\rho},\cal{A}_t])\right|+\nonumber\\&\frac{1}{2}\left|\tr([\sqrt{\rho},\cal{A}_t]^{\dagger}[\sqrt{\rho},\dot{\cal{A}_t}])\right|.
\end{align}
Let us apply the Cauchy--Schwarz inequality on the right hand side of the above inequality, we get 
\begin{equation}
     \left|\frac{{\rm d} }{{\rm d}t}I(\rho,\cal{A}_t)\right| \leq \sqrt{2}\norm{[\sqrt{\rho},\cal{L}^{\dagger}({\cal{A}}_t)]}_{\rm HS}\sqrt{I(\rho,\cal{A}_t)}.
\end{equation}
From above inequality we obtain
\begin{equation}
     \left|\int_0^{T}\frac{{{\rm d} I(\rho,\cal{A}_t)} }{\sqrt{I(\rho,\cal{A}_t)}}\right| \leq \sqrt{2}\int_{0}^{T}\norm{[\sqrt{\rho},\cal{L}^{\dagger}({\cal{A}}_t)]}_{\rm HS}{\rm d}t.
\end{equation}

After integrating above inequality, we obtain the following bound 
\begin{equation}
T\geq T_{Q}= \frac{\sqrt{2}|\sqrt{I(\rho,\cal{A}_T)}-\sqrt{I(\rho,\cal{A}_0)}|}{\langle\!\langle\norm{[\sqrt{\rho},\cal{L}^{\dagger}({\cal{A}}_t)]}_{\rm HS}\rangle\!\rangle_T}.
\end{equation}

\end{proof}
While deriving the above theorem we have utilized Heisenberg picture of quantum mechanics where observable evolves in time and state remains fixed. It is important to note that for non-unitary dynamics $I(\rho_t,A)\neq I(\rho,A_t)$.

\section{Quantumness for unitary evolution}\label{Appendix:unitary}
Let $\mathcal{A}_{0}=\hat{n}\cdot\Vec{\sigma}$, be unitarily evolved through $H=\hat{m}\cdot\Vec{\sigma}$. With $U(t)=e^{-i H t}= \openone \cos t - i (\hat{m}\cdot\vec{\sigma}) \sin t$, we have
\begin{align}
   \mathcal{A}_{t} =&\hspace{0.2cm} U(t)\mathcal{A}_{0} U(t)^{\dagger}\nonumber\\
    =& \cos (2t) \hat{n}\cdot\Vec{\sigma} - \sin(2t)(\hat{m}\times\hat{n})\cdot\Vec{\sigma} \nonumber\\
    &+ 2(\cos\theta \sin^2 t) \hat{m}\cdot\vec{\sigma}
\end{align}
Where $\cos\theta = \hat{m}\cdot\hat{n}$. In terms of these variables the terms appearing in $T_{QSL}$ can be computed as follows:

\begin{align}
    \mathcal{L}(\mathcal{A}_{t}) =& \hspace{0.1cm}\dot{\mathcal{A}_{t}} = i[H,\mathcal{A}_{t}]\nonumber\\
    =& -2[\cos(2t)(\hat{m}\times \hat{n})\cdot \Vec{\sigma} - \sin(2t)\cos\theta \hat{m}\cdot\Vec{\sigma}  \nonumber\\
    &+ \sin(2t) \hat{n}\cdot\Vec{\sigma}]
\end{align}

\begin{align}
    Q(\mathcal{A}_{0},\mathcal{A}_{t}) &= 2\tr\left([\mathcal{A}_{0},\mathcal{A}_{t}]^{\dagger}[\mathcal{A}_{0},\mathcal{A}_{t}]\right) \nonumber\\
    &= 8\left(\sin^2 {2t}\sin^2{\theta} + \sin^4{t} \sin^2{2\theta} \right)
\end{align}

\begin{eqnarray}
    \left[\mathcal{A}_0, \mathcal{L}^{\dag}\left(\mathcal{A}_t\right)\right]= -4 i[(\cos 2 t ) \hat{m} \cdot \vec{\sigma}-(\cos 2 t \cos \theta) \hat{n} \cdot \vec{\sigma} 
- \sin 2 t \cos \theta(\hat{n} \times \hat{m}) \cdot \vec{\sigma}]
\end{eqnarray}

\begin{eqnarray}
    \left\|\left[\mathcal{A}_0, \mathcal{L}^{\dag}\left(\mathcal{A}_t\right)\right]\right\|_{H  S}&=&\sqrt{\operatorname{Tr}\left(\left[\mathcal{A}_0, \mathcal{L}^{\dag}\left(\mathcal{A}_t\right)\right]^{\dag}\left[\mathcal{A}_0,\mathcal{L}^{\dag}\left(\mathcal{A}_t\right)\right]\right)} \nonumber\\
    &=& 4 \sqrt{2\left(\cos ^2 2 t+\sin ^2 2 t \cos ^2 \theta \sin ^2 \theta\right) }
\end{eqnarray}

Plugging these terms into Eq~\eqref{QQSL}, we obtain
\begin{equation}
    T\geq T_{Q} = \frac{\sqrt{\left(\sin^2 {2T}\sin^2{\theta} + \sin^4{T} \sin^2{2\theta} \right)}}{\frac{2\sqrt{2}}{T}\int_{0}^{T}dt  \sqrt{\left(\cos ^2 2 t+\sin ^2 2 t \cos ^2 \theta \sin ^2 \theta\right) }}
\end{equation}

\section{Evolution of observables for General dephasing process}\label{Appendix:dephasing}

To solve the master equation, we vectorize the operators $\mathcal{A}$ as follows. 

\begin{eqnarray}
    \mathcal{A}= \sum_{i,j} a_{ij} \ketbra{i}{j} \rightarrow \ |\mathcal{A})= \sum_{i,j} a_{ij} |ji)
\end{eqnarray}

The vectorization of the product of three matrices $A,B,C$ are given by the relation $|ABC)=(C^{T}\otimes A) |B)$. Using this, Eq. \eqref{heismastereqn} can be written as:
\begin{eqnarray}\label{vecmaster}
    \frac{{\rm d}}{{\rm dt}}|\mathcal{A}_t)&=& \left[i(\openone\otimes H - H^T \otimes \openone)-\frac{\gamma}{2}(\openone\otimes\openone -\sigma^T_z \otimes \sigma_z)\right]|\mathcal{A}_t)\nonumber\\
    &\equiv& M_{\mathcal{L}} |\mathcal{A})
\end{eqnarray}

Since $M_{\mathcal{L}}$ is time independent, the solution to Eq. \eqref{vecmaster} is given by:

\begin{eqnarray}\label{At}
    |\mathcal{A}_t)= e^{t M_{\mathcal{L}}} |\mathcal{A}_0).
\end{eqnarray}

\end{widetext}

\end{document}